%% file: mitl-mc-concur-arxiv-v2.tex
\documentclass[a4paper,UKenglish,thm-restate]{lipics-v2021}

\usepackage{wrapfig}
\usepackage{xxcolor}
\usepackage{bussproofs}
\usepackage{logicproof}[2014/03/20]

\ifpdf
  \usepackage{tikz}
  \usetikzlibrary{decorations.pathreplacing,calligraphy}
  \usetikzlibrary{arrows,automata}
  \usepackage[textsize=small]{todonotes}
  
\else
  \AtBeginDocument{%
  }
\fi

\newcommand{\Output}{\textsf{output} }
\newcommand{\goto}{\textsf{goto} }

\usepackage{algorithm}
\usepackage{algpseudocode}
\algblockdefx[MATCH]{Match}{EndMatch}%
  [1]{\textbf{match} #1 \textbf{with}}%
  {\textbf{end match}}
\algdef{SxN}[CASE]{Case}{EndCase}%
  [1]{{#1}:}
\algdef{SxN}[IFTHEN]{IfThen}{EndIfThen}%
  [2]{\textbf{if}~{#1}~\textbf{then} {#2}~\textbf{end if}}
\algdef{SxN}[FORONELINE]{ForOne}{EndForOne}%
  [2]{\textbf{for}~{#1}~\textbf{do} {#2~\textbf{end for}}}
\algblockdefx[CHOOSE]{Choose}{EndChoose}%
  [1]{\textbf{choose non-deterministically}}%
  {\textbf{end choose}}
\algdef{SxN}[WHEN]{When}{EndWhen}%
  [2]{\textbf{when}~{#1}~\textbf{do}~{#2}}

\input{declarations}

\title{MITL Model Checking via Generalized Timed Automata and a New Liveness Algorithm}

\titlerunning{Generalized Timed Automata: Liveness and MITL Model Checking}

\author{S.\ Akshay}{Department of CSE, Indian Institute of Technology Bombay, Mumbai, India}
  {akshayss@cse.iitb.ac.in}{https://orcid.org/0000-0002-2471-5997}
  {Supported in part by DST/SERB Matrics Grant MTR/2018/000744.}

\author{Paul Gastin}{Université Paris-Saclay, ENS Paris-Saclay, CNRS, LMF, 91190,
	Gif-sur-Yvette, France \and CNRS, ReLaX, IRL 2000, Siruseri, India}
  {paul.gastin@ens-paris-saclay.fr}{https://orcid.org/0000-0002-1313-7722}{}

  \author{R. Govind}{Uppsala University, Sweden}{govind.rajanbabu@it.uu.se}{https://orcid.org/0000-0002-1634-5893}{}

\author{B.\ Srivathsan}{Chennai Mathematical Institute, India
  \and CNRS, ReLaX, IRL 2000, Siruseri, India} 
  {sri@cmi.ac.in}{https://orcid.org/0000-0003-2666-0691}{}

\authorrunning{S. Akshay, P. Gastin, R. Govind and  B. Srivathsan}
\Copyright{S. Akshay, P. Gastin, R. Govind and  B. Srivathsan}

\keywords{MITL model checking, timed automata, zones, liveness }

\category{}
\ccsdesc[500]{Theory of computation~Timed and hybrid models}
\ccsdesc[300]{Theory of computation~Quantitative automata}
\ccsdesc[100]{Theory of computation~Logic and verification}

\nolinenumbers
\hideLIPIcs  

\begin{document}

\maketitle

\begin{abstract}
  The translation of Metric Interval Temporal Logic (MITL) to timed
  automata is a topic that has been extensively studied. A key challenge here is the
  conversion of future modalities into equivalent automata. Typical
  conversions equip the automata with a guess-and-check mechanism to
  ascertain the truth of future modalities.  
  Guess-and-check can be naturally implemented via alternation. However, since timed automata tools do
  not handle alternation, existing methods perform an additional
  step of converting the alternating timed automata into 
  timed automata. This ``de-alternation'' step proceeds by an
  intricate finite abstraction of the space of configurations of the
  alternating automaton. 

  Recently, a model of generalized timed automata (GTA) has been
  proposed. The model comes with several powerful additional features,
  and yet, the best known zone-based reachability algorithms for timed
  automata have been extended to the GTA model, with the same complexity
  for all the zone operations. An attractive
  feature of GTAs is the presence of future clocks which act like
  timers that guess a time to an event and stay alive until a
  timeout. Future clocks seem to provide another natural way to
  implement the guess-and-check: start the future clock with a guessed time to an event and check its occurrence using a timeout. Indeed, using this feature, we provide a new concise translation from MITL to GTA. In particular,
  for the timed until modality, our translation offers an exponential 
  improvement w.r.t. the state-of-the-art.

  Thanks to this conversion, MITL model
  checking reduces to checking liveness for GTAs. However, no
  liveness algorithm is known for GTAs. Due to the presence of future clocks,
  there is no finite time-abstract bisimulation (region equivalence)
  for GTAs, whereas liveness
  algorithms for timed automata crucially rely on the presence of the
  finite region equivalence. As our second contribution, we
  provide a new zone-based algorithm for checking Büchi
  non-emptiness in GTAs, which circumvents this fundamental
  challenge. 
\end{abstract}

\section{Introduction}

The translation of Linear Temporal Logic (LTL)~\cite{Pnueli77} to
Büchi automata is a fundamental problem in model checking, with a
long history of theoretical
advances~\cite{GPVW95,vardi1996automata,GastinO01}, tool
implementations~\cite{Spin,Spot,GastinO01,Strix,NuSMV} and practical
applications~\cite{Pnueli86,PnueliH88,JensenLS96,Lamport2002}.  In the
real-time setting, Metric Interval Temporal Logic (MITL) is close to
LTL, with the modalities Next ($\X$) and Until ($\U$) extended with
timing intervals --- for instance, $X_{[a,b]}p$ says that the next
event is a $p$ and it occurs within a delay $\theta \in [a,b]$. Model
checking for MITL is known to be
$\mathsf{EXPSPACE}$-complete~\cite{AFH-MITL-J}. This has led to the
study of ``efficient'' conversions from MITL to timed automata, with
each new construction aiming to make the automata more succinct. Our
work is another step in this direction.

There are two ways to interpret MITL formulae: over (continuous) timed
signals~\cite{AFH-MITL-J,MalerNP06,FerrereMNP19} or (pointwise) timed
words~\cite{AlurH90,Wilke94,MightyL}. Since the current timed automata
tools work with timed words, we stick with the pointwise semantics.
The state-of-the-art for MITL-to-TA is based on an initial translation
of MITL to one-clock Alternating Timed Automata
(OCATA)~\cite{OuaknineW06}. It has been shown that these OCATA can be
converted to a network of timed
automata~\cite{BrihayeEG13,BrihayeEG14}. The tool
MightyL~\cite{MightyL} implements the entire MITL-to-TA translation.
One of the difficulties in the MITL-to-TA translation is the inherent
mismatch between the logic and the automaton in the way timing
constraints are enforced. A future modality declares that a certain
event takes place at a certain timing distance, \emph{in the future}.
In a timed automaton, \emph{clocks} measure time elapsed since some
event \emph{in the past} and check constraints on these values.  To
implement a future modality, the automaton needs to make a prediction
about the event and verify that the prediction is indeed true.
Therefore, each prediction typically resets a clock and stores a new
obligation in the state. The automaton needs to discharge these
obligations at the right times in the
future.

\begin{figure}[tbh]
  \centering
  \begin{minipage}[b]{40mm}
    \includegraphics[page=11,scale=0.9]{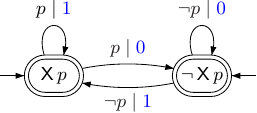}
    \\[3ex]
    \includegraphics[page=12,scale=0.9]{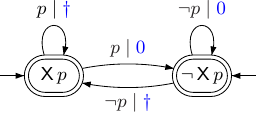}
    \\[1.5ex]
  \end{minipage}
  \hfil
  \includegraphics[page=13,scale=0.9]{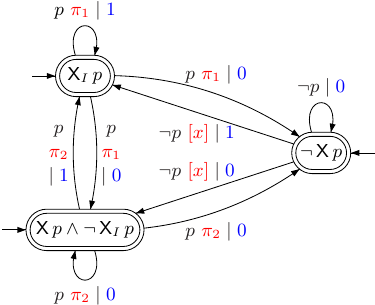}
  \caption{(top left) Büchi transducer (with outputs) for LTL formula $\X p$
    (right) Timed transducer with clock $x$ for MITL formula $\X_I p$;
    $\textcolor{red}{\pi_1 := x \in I; [x]}$, 
    $\textcolor{red}{\pi_2 := x \notin I; [x]}$
    (bottom left) a hypothetical transducer with a variable $\theta$ that predicts time to
    next action; $\textcolor{blue}{\dagger := \ite{\theta \in I}{1}{0}}$.}
  \label{fig:intro-example}
\end{figure}

Figure~\ref{fig:intro-example} (top left) shows an automaton with outputs for
the LTL formula $\X p$. On an infinite word $w_1 w_2 \dots $ (where
each $w_i$ is a subset of atomic propositions) the automaton outputs
$\textcolor{blue}{1}$ at $w_i$ iff $w_{i+1}$ contains $p$. While
reading $w_i$, the automaton needs to guess whether $p \in w_{i+1}$ or
not. Depending on the guess, it stores an appropriate obligation. This
is reflected in the states and transitions: transitions with output
$\textcolor{blue}{1}$ go to a state $\X p$ which can only read $p$
next, whereas those with output $\textcolor{blue}{0}$ go to $\neg\X p$
which can only read $\neg p$. The $\X p$ and $\neg \X p$ can be seen
as obligations that the automaton has to discharge from the state.

Now, let us consider a timed version $\X_{I} p$ interpreted on timed
words. An automaton for $\X_{I}p$ needs to guess whether the next
letter is a $p$ \emph{and} if so, whether it appears within $\theta$
time units for some $\theta \in I$. Figure~\ref{fig:intro-example} (bottom left)
represents a hypothetical automaton that implements this idea:
assuming it has access to a variable $\theta$ which contains the time
to the next event, the output should depend on whether $\theta \in I$
or not.  This is exactly what the if-then-else condition $\textcolor{blue}{\dagger}$ does:
if $\theta \in I$ output $1$, else output $0$. 
Classical timed automata do not have direct access to
$\theta$. They implement this idea differently, by making use of extra
states. Figure~\ref{fig:intro-example} (right) shows a timed automaton for
$\X_I p$. The state $\X p$ is
split into two different obligations: $\X_I p$ where the timing
constraint is satisfied, and $\X p \land \neg \X_I p$ where it is
not. The outgoing guards discharge these obligations. This example
shows the convenience of having access to a variable that can predict
time to future events.

This is precisely where the recently proposed model of
\emph{Generalized Timed Automata} (\GTA)~\cite{AkshayGGJS23} enters
the picture. This model subsumes event-clock automata~\cite{AlurFH99}
and automata with timers~\cite{Dill89}. $\GTA$ come equipped with the
additional resources to implement predictions better.  $\GTA$ have two
types of clocks: \emph{history clocks} and \emph{future clocks}.
History clocks are similar to the usual clocks
of timed automata. Future clocks are like timers, but instead of
starting them at some non-negative value and making them go down to
zero, they get started with some arbitrary negative value and go up
until they hit zero. For example, in Figure~\ref{fig:intro-example}
(bottom left), each transition can start a future clock,
guessing the time to the next event. This immediately gives us the
required $\theta$. The exact $\GTA$ for $\X_I p$ is quite close to
Figure~\ref{fig:intro-example} (bottom left) and is given in
Figure~\ref{fig:Next-timed}.  Apart from the use of future clocks, the
syntax of transitions in a GTA is much richer than a guard-reset pair
as in timed automata. Transitions contain an ``instantaneous timed
program'', which consists of a sequence of guards, resets and releases
(for future clocks).  When difference constraints are present, the
model becomes powerful enough to encode counter machines and is
therefore undecidable. A safe fragment, with a careful use of diagonal
constraints is known to be decidable.

$\GTA$s are advantageous in another sense. In spite of the powerful
features, the best zone-based algorithms from the timed automata
literature have been shown to suitably adapt to the $\GTA$ setting,
with the same complexity for zone operations, and have been
implemented in the tool TChecker~\cite{TChecker}.  Therefore, an MITL
to $\GTA$ conversion allows us to capitalize on the features and
succinct syntax of $\GTA$, and at the same time, lets us model check
MITL directly on richer $\GTA$ models. In
summary: 

\begin{itemize}
\item We provide a translation of MITL formulae
  to safe $\GTA$. The translation is compositional and implementable,
  and yields an \emph{exponential improvement} in the number of
  locations compared to the state-of-the-art technique for pointwise
  semantics~\cite{MightyL}, while the number of clocks remains the
  same up to a constant.

\item Model checking MITL against $\GTA$ requires to solve the
  liveness problem for (safe) $\GTA$s, which has been open so far. We
  settle the liveness question in this work. Zone based algorithms for
  event-clock automata have been studied in~\cite{GeeraertsRS14}. A
  notion of weak regions has been developed and this can be used for
  solving both reachability and liveness using zones. The GTA model
  that we consider in this paper strictly subsumes event-clock
  automata. In particular, the presence of diagonal constraints makes
  the problem more challenging. Our solution to liveness for GTAs
  therefore gives an alternate liveness procedure for event-clock
  automata, and also settles liveness for event-clock automata with
  diagonal constraints, a model defined
  in~\cite{DBLP:journals/tcs/BozzelliMP22}.
\end{itemize}

We remark that the techniques used in continuous semantics do not
extend to pointwise semantics.  In~\cite{FerrereMNP19} the authors
simplify general MITL formulae into formulae containing only one-sided
intervals (of the form $[0,c]$ or $[c,\infty)$), for which automata
are considerably simpler to construct.  However, this simplification
at the formula level works \emph{only} in the continuous semantics ---
it does not work in the pointwise-semantics (as Lemma 4.3, 4.4
of~\cite{FerrereMNP19} do not extend to pointwise-semantics).  The
fundamental difference is that in the continuous semantics we can
assert a formula at any time point $t$. However, in
pointwise-semantics, we can evaluate a formula only at \emph{action
  points}, i.e., points where there is an actual action. For example,
in continuous semantics one can rewrite $\F_{[a+c, b+c]} p$ as
$\F_{[0,c]} \G_{[0,c]} \F_{[a, b]} p$ when $c \leq b-a$ (Lemma
4.3,~\cite{FerrereMNP19}). However, in the pointwise semantics there
may be no event in the interval $[0,c]$ on which we can evaluate
$\G_{[0,c]} \F_{[a, b]} p$. Therefore, we need a completely different
approach to deal with intervals in the pointwise semantics.

\subparagraph*{Organization of the paper.} We start with preliminary
definitions of Generalized Timed Automata (Section~\ref{sec:prelims})
and provide our solution to the liveness problem in
Section~\ref{sec:gta-liveness}. We present our MITL to GTA translation
in Section~\ref{sec:mitl-gta}. 

\section{Preliminaries}
\label{sec:prelims}

Let $X = X_F \uplus X_H$ be a set of real-valued variables called
\emph{clocks}, which is further partitioned into \emph{future clocks}
$X_F$ and \emph{history clocks} $X_H$. Let
$\Phi(X)$ denote a set of \emph{clock constraints} generated by the
grammar: $\varphi ::= x - y \leqlt c \mid \varphi \land \varphi$ where
$x,y \in X \cup \{0\}$, ${\leqlt}\in\{<,\leq\}$ and
$c \in \overline{\mathbb{Z}}=\mathbb{Z}\cup\{-\infty,+\infty\}$ (the
set of integers equipped with the two special values to say that a
clock is ``undefined''). We also allow \emph{renamings} of clocks.  Let $\perm_{X}$ be the set of
permutations $\sigma$ over $X\cup\{0\}$ mapping history (resp.\ future) clocks to history
(resp.\ future) clocks ($\sigma(X_{F})=X_{F}$ and $\sigma(X_{H})=X_{H}$).

\smallskip

\noindent \textbf{GTA syntax.} A \emph{Generalized Timed Automaton (GTA)} is given by 
$(Q, \Sigma, X, \Delta, \Ii, Q_{f})$
where $Q$ is a finite set of states, 
$\Sigma$ is a finite alphabet of actions,
$X = X_F \uplus X_H$ is a set of clocks partitioned into future clocks $X_F$ and
history clocks $X_H$. 
The initialization condition $\Ii$ is a set of pairs $(q_0, g_0)$ where a pair consists of
an initial state $q_0 \in Q$ and an initial guard $g_0 \in \Phi(X)$, 
and the accepting condition is given by a set $Q_{f}\subseteq Q$ of Büchi states.
The transition relation $\Delta \subseteq (Q \times \Sigma \times \Prog \times Q)$ contains transitions of the form $(q, a, \prog, q')$, where $q$ is the
source state, $q'$ is the target state, $a$ is the action triggering
the transition, and $\prog$ is an \emph{instantaneous timed program}
generated by the grammar: 
 \begin{align*}
   \prog := \guard \mid \chg \mid \rnm \mid \prog; \prog
 \end{align*}
where $\guard =g\in \Phi(X)$, $\chg=[R]$ for an $R\subseteq X$, and $\rnm=[\sigma]$ for a
$\sigma\in\perm_{X}$.

Figure~\ref{fig:Next-timed} with the blue parts removed illustrates a GTA. Both states
$\ell_1$ and $\ell_2$ are initial, denoted by incoming arrows to each of them, and
accepting, marked by the double circle.  The initial guard is the trivial \emph{true}
constraint.  The alphabet $\Sigma = \{0 , 1\}$ (written in black).  The constraint $-x \in
I$ is short form for a conjunction of constraints requiring the clock to be in the
interval $I$.  For example, if $I = (4, 5]$, then $-x \in I$ is the constraint $4 < -x
\land -x \le 5$.  During our MITL to GTA translation, we extend GTAs to include outputs (Definition~\ref{defn:gtt}. 
The dagger condition
\textcolor{blue}{$\ite{(-x\in I)}{1}{0}$} is a short form for two transitions, one which
checks $-x \in I$ and outputs $1$, and the other which checks $-x \notin I$ and outputs
$0$.

\smallskip

\noindent \textbf{GTA semantics.} A valuation of clocks is a function
$v\colon X \cup \{0\} \mapsto \overline{\RR}=\mathbb{R}\cup
\{-\infty,+\infty\}$ which maps the special clock $0$ to 0, history
clocks to $\Rpos\cup\{+\infty\}$ and future clocks to
$\Rneg \cup \{-\infty\}$.  We denote by $\V(X)$ or simply by $\V$ the
set of valuations over $X$. For a valuation $v\in\V$,
define\footnote{To allow evaluation of all the constraints in
  $\Phi(X)$, the addition and the unary minus operation on real
  numbers is extended~\cite{AkshayGGJS23} with the following
  conventions (i)$(+\infty)+\alpha = \alpha+(+\infty) = +\infty$ for
  all $\alpha\in\overline{\RR}$, (ii)
  $(-\infty)+\beta = \beta+(-\infty) = -\infty$, as long as
  $\beta\neq+\infty$, and (iii) $-(+\infty)=-\infty$ and
  $-(-\infty)=+\infty$. }  $v\models y-x \leqlt c$ as
$v(y)-v(x)\leqlt c$. We say that $v$ \emph{satisfies} a constraint
$\varphi$, denoted as $v\models\varphi$, when $v$ satisfies all the
atomic constraints in $\varphi$.
We denote by $v + \d$ the \emph{valuation} obtained from valuation $v$
by increasing by $\d \in \Rpos$ the value of all clocks in $X$.  Note
that, from a given valuation, not all time elapses result in
valuations since future clocks need to stay at most $0$.
We now define the \emph{change} operation that combines the
\emph{reset} operation for history clocks (which sets history clocks to 0) and
\emph{release} operation for future clocks (which assigns a
non-deterministic value to a future clock).  Given a set of clocks $R \subseteq X$, we
define $R_{F} = R \cap X_{F}$ as the set of future clocks in $R$, and
$R_{H} = R \cap X_{H}$ as the set of history clocks in $R$.  Then,
$[R]v := \{v' \in \V \mid v'(x) = 0~\forall ~x \in R_{H} \text{ and }
v'(x) = v(x) ~\forall~ x \not\in R\}$. Observe that $v'$ has no
constraints for the future clocks in $R$, as they can take any arbitrary value in $v'$. 
For a valuation $v\in\V(X)$, and $\sigma \in \perm_X$, we define $[\sigma]v$ as
$v\circ\sigma$, i.e., $([\sigma]v)(x)=v(\sigma(x))$ for all $x\in X\cup\{0\}$.

For valuations $v, v'$ and a guard $g \in \Phi(X)$ we write
$v\xrightarrow{g}v'$ when $v'=v\models g$, and
$v\xrightarrow{[R]}v'$ when $R\subseteq X$ and $v'\in[R]v$, and $v\xrightarrow{[\sigma]}v'$ when $\sigma\in\perm_{X}$ and $v'=[\sigma]v$. When
$\prog=\prog_1;\ldots;\prog_n$, we write
$v\xrightarrow{\scriptstyle\prog}v'$ when there are valuations
$v_1,\ldots,v_n$ such that
$v\xrightarrow{\scriptstyle \prog_1}v_1 \xrightarrow{\scriptstyle
  \prog_2} \cdots \xrightarrow{\scriptstyle \prog_n} v_n=v'$.  The
semantics of the \GTA\ $\Aa$ defined above is given by a transition
system $\TS_{\Tt}$ whose states are \emph{configurations} $(q,v)$ of
$\Aa$, where $q \in Q$ and $v\in\V$ is a valuation.  A configuration
$(q,v)$ is initial if $v\models g$ for some $(q,g) \in \Ii$, and it is
accepting if $q\in Q_{f}$. 
Transitions of $\TS_{\Tt}$ are of two forms: (1) \textit{delay transition}:
$(q,v) \xra{\delta} (q, v + \delta)$ if $v+\delta$ is a valuation,
i.e., $(v + \delta) \models X_{F} \leq 0$, and (2) \textit{discrete
  transition}: $(q,v) \xra{t} (q',v')$ if $t=(q,a,\prog,q')\in\Delta$
and $v\xrightarrow{\prog} v'$.  A \emph{finite (respectively infinite)
  run} $\rho$ of a \GTA\ is a finite (respectively infinite) sequence
of transitions from an initial configuration of $\TS_{\A}$:
$(q_0, v_0) \xrightarrow{\delta_0, t_0} (q_1, v_1)
\xrightarrow{\delta_1, t_1} \cdots$.

For example, consider the run of GTA in Figure~\ref{fig:Next-timed} on
a timed word $(1, 1) (0, 2) (1, 3) (0, 4) \dots$ ($1$ occurs
at all odd numbers, and $0$ at all even numbers, starting from first
timestamp $1$). The program $(x = 0); [x]$ used in the
transitions first checks if $x$ is $0$, and then releases
it to an arbitrary non-deterministic value. The run on the above word would be:
$(\ell_1, x=-1) \xra{1, t_{1}} (\ell_2, x=-1) \xra{1, t_{2}} (\ell_1, x=-1) \cdots$.  A
transition $\xra{\delta, t}$ denotes a time elapse of $\delta$ followed by application of the
program associated to transition $t$.  At each point the value of $x$ is released to
$-1$, and is checked with the guard $x = 0$ at the next event.

An infinite run is accepting if it visits accepting configurations infinitely often.  The
run is said to be \emph{Zeno} if $\Sigma_{i \ge 0} \delta_i$ is bounded and
\emph{non-Zeno} otherwise.  In this work, we will only be interested in non-Zeno runs (we
can force non-Zenoness by ensuring time elapse in every loop of the
automaton~\cite{TripakisYB05}).

\smallskip

\noindent \textbf{Liveness problem.} The \emph{non-emptiness or liveness problem} for a
\GTA\ asks whether the given \GTA\ has an accepting non-Zeno run.  Unfortunately, the
non-emptiness problem even for finite words turns out to be undecidable for general
\GTA~\cite{AkshayGGJS23}.  Therefore, we focus our attention on a restricted sub-class of
\GTA's for which non-emptiness in the finite words case is decidable, called safe
\GTA~\cite{AkshayGGJS23}.

\begin{definition}[Safe \GTA~\cite{AkshayGGJS23}]\label{defn:safe-gta}
  \label{def:safe-program}
  Given a \GTA\ $\Aa$, let $X_D\subseteq X_F$ be the subset of future clocks used in
  diagonal guards of $\Aa$ between future clocks, i.e., if $x-y\leqlt c$ with $x,y\in
  X_{F}$ occurs in some guard of $\Aa$ then $x,y\in X_{D}$.  
  Then, a program $\prog$ is $X_D$-safe if
  clocks in $X_{D}$ are checked for being $0$ or $-\infty$ before being released 
  and renamings $[\sigma]$ used in $\prog$ preserve $X_{D}$ clocks 
  ($\sigma(X_{D})=X_{D}$).
  A \GTA\ $\Aa$ is safe if it only uses $X_{D}$-safe programs on its transitions and the
  initial guard $g_0$ sets each history clock to either $0$ or $\infty$.
\end{definition}

The GTA in Figure~\ref{fig:Next-timed} is vacuously safe, since
  there are no diagonal constraints at all. 

\begin{remark}\label{rem:renaming}
  Renaming operations may be considered as syntactic sugar allowing for more concise
  representations of \GTA s.  Indeed, we can transform a \GTA\ $\Aa$ with renamings to an
  equivalent \GTA\ $\Aa'$ without renamings by adding to the state the current permutation
  of clocks (composition of the permutations applied since the initial state) and change
  the programs of outgoing transitions accordingly.  The number of states is multiplied by
  the number of permutations that may occur as described above.
\end{remark}

\noindent \textbf{Zones for \GTAFull .}\label{sec:liveness-prelims}
The state-of-the-art approach for answering the B\"uchi non-emptiness problem and the reachability problem for timed automata proceeds via checking reachability in a graph called the \emph{zone graph} of a timed automaton~\cite{Daws}.
Abstractly, \emph{zones}~\cite{Bengtsson:LCPN:2003} are essentially sets of valuations that can be represented efficiently using constraints between differences of clocks. Here, we recall the notion of zones for the setting of \GTAfull .

\begin{definition}[\textbf{\GTA\ zones}~\cite{AkshayGGJS23}]\label{defn:gta-zones}
  A \GTA\ zone is a set of valuations satisfying a conjunction of constraints of the form
  $y-x \leqlt c$, where $x,y \in X\cup\{0\}$,
  $c\in\overline{\mathbb{Z}}$ and
  ${\leqlt}\in\{\leq,<\}$.
\end{definition}
Thus zones are an abstract representation of sets of valuations. 

Let $W$ be an arbitrary set of valuations (not necessarily a \GTA\ zone) and $q$ be a state. 
For a transition $t:= (q, a, \prog, q')$, we write $(q, W)$ $\xra{t}$ $(q', W_t)$ if
$W_t = \{ v' \mid \text{$v'$ is a valuation and } (q, v) \xra{t} \xra{\delta} (q', v')
\text{ for some } v \in W \text{ and } \delta \in \mathbb{R}_{\geq 0} \}$.
\cite{AkshayGGJS23} showed that starting from a \GTA\ zone $Z$, the successors
are also event zones: $(q, Z) \xra{t} (q', Z_t)$ implies $Z_t$ is also a a \GTA\ zone.
This observation is used to define the notion of the \emph{Zone graph} of a \GTAfull .

\begin{definition}[\textbf{\GTA\ zone graph}~\cite{AkshayGGJS23}]
  Given a \GTA\ $\Aa$, its {\em \GTA\ zone graph, denoted \gzg($\Aa$)}, is defined as follows: 
  Nodes are of the form $(q, Z)$ where $q$ is a state and $Z$ is a \GTA\ zone.  
  Initial nodes are pairs $(q_0, \elapse{Z_0})$ where $(q_0,g_{0})\in\Ii$ is an initial
  condition and $Z_0$ is
  given by $g_{0}\wedge\big(X_{F} \leq 0\big)\wedge \big(X_H \geq 0\big)$
  ($Z_{0}$ is the set of all valuations which satisfy the \emph{initial constraint} $g_0$).
  For every node $(q, Z)$ and every transition $t := (q, a, \prog, q')$ there is a
  transition $(q, Z) \xra{t} (q', Z')$ in the \GTA\ zone graph.
\end{definition}

Finally, just as is the case for zone graphs for timed automata, $\gzg(\Aa)$ is not
guaranteed to be finite.  In order to use it to check Büchi non-emptiness or reachability,
we need a finite abstraction of the zone graph.  The standard technique to obtain such
finite abstractions is using the notion of \emph{simulations}, that we recall next.

\begin{definition}[\textbf{Simulation}]\label{def:simulation}
  A (time-abstract) simulation relation on the semantics of a \GTA\ is a reflexive,
  transitive relation $(q,v) \preceq (q,v')$ relating configurations
  with the same control state and 
  \begin{enumerate}
    \item for every $\delta\in\Rpos$ such that $v+\delta\in\V$ is a valuation, there
    exists $\delta'\in\Rpos$ such that $v'+\delta'\in\V$ is a valuation and
    $(q,v+\delta)\preceq(q,v'+\delta')$,
  
    \item for every transition $t$, if $(q,v) \xra{t} (q_1,v_1)$ for some valuation
    $v_{1}$, then $(q,v') \xra{t} (q_1,v'_1)$ for some valuation $v'_{1}$ with
    $(q_1,v_1)\preceq(q_1,v'_1)$,
    
    \item for all future clock $x\in X_{F}$, if $v(x)=-\infty$ then 
    $v'(x)=-\infty$.
  \end{enumerate}
  For two \GTA\ zones $Z, Z'$, we say $(q, Z) \preceq (q, Z')$ if for
  every $v \in Z$ there exists $v' \in Z'$ such that
  $(q, v) \preceq (q, v')$. 
\end{definition}

\section{Liveness for GTA}
\label{sec:gta-liveness}

In this section, we will discuss a zone-based procedure to check
liveness for safe generalized timed automata. We start by explaining
how the standard zone based algorithm for solving liveness in classical
timed automata can be adapted to the setting of safe \GTA s. The
approach for timed automata crucially depends on the existence of a
finite time-abstract bisimulation between valuations, namely the
region-equivalence~\cite{Alur:TCS:1994}. However, there exists no such
finite time-abstract bisimulation for \GTA s (extension of a result of
\cite{GeeraertsRS14}), as illustrated in Figure~\ref{fig:no-finite}.
The issue is that we cannot forget (abstract) the values of future clocks, unlike history
clocks where values above a maximum constant are equivalent.  Therefore, our approach
involves a significant deviation from the standard one.  

\begin{figure}[tbh]
  \centering
  \includegraphics[page=9,scale=0.9]{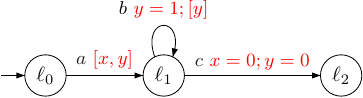}
  \caption{Example to illustrate no finite bisimulation in GTA; $x$ is
    a future clock, $y$ a history clock.
    The initial transition releases clock $x$ to an arbitrary value, and resets $y$ to
    $0$.  From configuration $\langle \ell_1, x=-n, y=0 \rangle$ ($n\in\mathbb{N}$), 
    the only way to reach $\ell_2$ is by executing $b^nc$, with $1$ time unit
    between consecutive $b$'s.
    Therefore, $\langle \ell_1, x = -n, y = 0 \rangle$ and 
    $\langle \ell_1, x = -m, y = 0 \rangle$ are simulation incomparable, 
    when $n \neq m$.  Hence there is no finite  bisimulation.}
  \label{fig:no-finite}
\end{figure}

We fix a safe \GTA\ $\A$ for the rest of this section. In order to
focus on the main difficulties and avoid additional technicalities, we
assume that the \GTA\ $\Aa$ is without renamings. We start by noting
that non-Zeno runs have a special form: future clocks which are not
ultimately $-\infty$ should be released infinitely often.  If not,
there is a last point where a future clock is released to a finite
value, and the entire suffix of the run should fall under this finite
time, which contradicts non-Zenoness.
\begin{restatable}{lemma}{lemnonzeno}
  \label{lem:non-zeno}
  Let
  $\rho := (q_0,v_0) \xra{\delta_0,t_0} (q_1,v_1) \xra{\delta_1,t_1}
  \cdots$ be a non-Zeno run of the \GTA\ $\A$.  Then, for every future
  clock $x$ of $\A$, and for every index $i \ge 0$, if
  $v_i(x) \neq -\infty$, there exists $j \geq i$ such that $x$ is
  released in $t_j$.
\end{restatable}

\begin{proof}
  Suppose not. There is a future clock $x$, and an index $i$ such that
  $v_i(x) = \alpha \neq -\infty$, and $x$ is never released in any
  later transition. By the semantics of delay transitions,
  $\sum_{j \ge i} \delta_j \le \alpha$. This contradicts the non-Zenoness
  assumption.
\end{proof}

\subparagraph*{\textbf{Overview of our solution.}}  

In classical timed automata, the liveness problem is solved by enumerating the zone
graph~\cite{Daws,Bengtsson:LCPN:2003},
and using a \emph{simulation
equivalence}~\cite{HerbreteauSW12,BBLP-STTT05,GastinMS18,Gastin0S19} for termination:
nodes of a zone graph are of the form $(q, Z)$ where $q$ is a state of the automaton, and
$Z$ is a zone; exploration from $(q, Z)$ is stopped if there exists an already visited
node $(q, Z')$ such that $(q,Z)\preceq(q,Z')$ and $(q,Z')\preceq(q,Z)$ for some simulation
relation $\preceq$.  In this case an edge is added between $(q, Z)$ and $(q, Z')$ to
indicate a simulation equivalence.  There is an (infinite) accepting run iff there is a
cycle in the zone graph thus computed, containing an accepting state.  The main point is
that, from a cycle in the zone graph with simulation equivalences, we can conclude the
existence of an infinite run over configurations.

At a high level the proof of this fact is as follows. Let us start by
ignoring simulations for the moment: suppose
$(q, Z) \xra{\sigma} (q, Z)$ for a sequence of transitions $\sigma$. By
definition of successor computation in the zone graph, for every $v$
in the zone $Z$ (on the right), there exists a predecessor $u$ in the
zone $Z$ (on the left). Repeatedly applying this argument gives a
valuation $u \in Z$ from which $\sigma$ can be iterated $\ell$ times, for
any $\ell \ge 1$. When $\ell$ is greater than the number of Alur-Dill
regions~\cite{Alur:TCS:1994}, we get a run
$(q, u) \xra{\sigma^{\ell}} (q, u')$ such that $u$ and $u'$ are region
equivalent. Since the region equivalence is a time-abstract
bisimulation, this shows that we can once again do $\sigma^{\ell}$
from $(q, u')$, and so on. This leads to an infinite run from $(q, u)$
where $\sigma$ can be iterated infinitely often. Now, when simulations
are involved, we need to consider sequences of the form
$(q, Z) \xra{\sigma} (q, Z')$ where $(q,Z)\preceq(q,Z')$ and $(q,Z')\preceq(q,Z)$.
An argument similar to the above can be adapted in this case
too~\cite{LiBuchi09,HerbreteauSW-cav-12,LivenessHard}. The critical
underlying reason that makes such an argument possible is the presence
of a finite time-abstract bisimulation, which in timed automata, is
given by the region equivalence.

The same idea cannot be directly applied in the GTA setting, as there
is no finite time-abstract bisimulation for GTAs, 
even with the safety assumption (Figure~\ref{fig:no-finite}). However, \cite{AkshayGGJS23} have defined a
finite equivalence $v_1 \sim_M v_2$ and shown that the downward
closures of 
the reachable zones w.r.t. a certain simulation called the $\G$
simulation~\cite{GastinMS18,Gastin0S19} are unions of $\sim_M$ equivalence
classes. Therefore, applying an argument of the above style will give
us a run $(q, u) \xra{\sigma^{\ell}} (q, u')$ such that $u \sim_M
u'$. But we cannot conclude an infinite run immediately as $\sim_M$ is
not a bisimulation.

To circumvent this problem, we will define an equivalence $\eqv_M$
which is in spirit like the region equivalence in timed automata. As
expected, $\eqv_M$ will be a bisimulation. However, in accordance with
the no finite timed-bisimulation result, $\eqv_M$ will have an
infinite index. We make a key observation: if we have a run
$(q, u) \xra{\sigma} (q, u')$ such that $u \sim_M u'$ \emph{and} if
$\sigma$ releases every future clock, then we can get a run
$(q, u) \xra{\sigma} (q, u'')$ where $u \eqv_M u''$ for a suitably
modified valuation $u''$. Since $\eqv_M$ is
a bisimulation, this will then give an infinite run where $\sigma$ can
be iterated infinitely often. As we have seen from
Lemma~\ref{lem:non-zeno}, if we are interested in non-Zeno runs, only
such cycles where all future clocks are released (or remain $-\infty$)
are relevant.
Therefore, in order to decide liveness for safe GTAs, it suffices to
construct the zone graph with the simulation replaced with the
simulation equivalence and look for a reachable cycle that contains an accepting
state such that for every future clock $x$, either $x$ is released on the cycle, or 
valuation $-\infty$ is possible for clock $x$.

This section is
organized as follows: we will first define the equivalence $\eqv_M$
and show that it is a bisimulation; then we recall $\sim_M$, and prove
the key observation mentioned above. One of the main challenges is in
addressing diagonal constraints, which is exactly where the safety
assumption is helpful.

\subparagraph*{\textbf{A region-like equivalence for \GTA.}}  

The definition of $\eqv_M$ looks like the classical region equivalence extended from $[0,
+\infty)$ to $\overline{\mathbb{R}}$: all clocks which are lesser than $M$ (which
automatically includes all future clocks) have the same integral values, and the ordering
of fractional parts among these clocks is preserved.  To account for diagonal constraints
in guards, we explicitly add a condition to say that all allowed diagonal constraints are
satisfied by equivalent valuations.  This new equivalence does not have a finite index,
but it turns out to be a time-abstract bisimulation, similar to the classical regions.

\begin{definition}\label{def:eqv_M}
  Let $v_{1},v_{2}\in\V$ be valuations.
  We say $v_1 \eqv_M v_2$ if for all clocks $x, y$:
  \begin{enumerate}
    \item $v_{1}(x)\leqlt c$ iff $v_{2}(x)\leqlt c$ for all ${\leqlt}\in\{<,\le\}$
    and $c\in\{-\infty,+\infty\}$ or $c\in\mathbb{Z}$ with $c\leq M$,
    
    \item $v_1\models x - y \leqlt c$ iff $v_2\models x - y \leqlt c$ for all
    ${\leqlt}\in\{<,\le\}$ and $c\in\{-\infty,+\infty\}$ or $c\in\mathbb{Z}$ with $|c|\le M$,
    
    \item if $-\infty<v_1(x),v_1(y)\le M$ then we have
    $\{v_1(x)\} \le \{v_1(y)\}$ iff $\{v_2(x)\} \le \{v_2(y)\}$.
  \end{enumerate}
\end{definition}

Notice that when $v_{1}\eqv_{M} v_{2}$, the first condition implies  
$v_{1}(x)=+\infty$ iff $v_{2}(x)=+\infty$,
$v_{1}(x)=-\infty$ iff $v_{2}(x)=-\infty$, $-\infty<v_1(x)\le M$ iff $-\infty<v_2(x)\le M$,
and in this case $\Int{v_1(x)}=\Int{v_2(x)}$ and $\{v_1(x)\}=0$ iff $\{v_2(x)\}=0$.

\begin{restatable}{lemma}{eqvisbisimulation}
  \label{lem:eqv-is-bisimulation}
  $\eqv_M$ is a time-abstract bisimulation.
\end{restatable}

\begin{proof}
  Let $v_1 \eqv_M v_2$. We will show that the equivalence is preserved
  under time delays and transitions.
  \begin{description}
  \item[Guard satisfaction.] We consider guards to have constants in
    $\{-\infty,+\infty\}\cup\{c\in\mathbb{Z}\mid |c|\leq M\}$.
    From Items $1,2$ we directly get $v_{1}\models g$ iff $v_{2}\models g$.

  \item[Time elapse.] Let $\delta_1 \ge 0$. 
    We will construct $\delta_2\geq0$  such that
    $v'_{1}=v_1 + \delta_1 \eqv_M v_2 + \delta_2=v'_{2}$. Set
    $\Int{\delta_2} = \Int{\delta_1}$. If $\{ \delta_1\} = 0$, then
    set $\{\delta_2\} = 0$. Now, assume $0 < \{ \delta_1\} < 1$.

    Let $B = \{ x \mid -\infty<v_1(x) \le M \}=\{ x \mid -\infty<v_2(x) \le M\}$.  Let
    $\{f_1, \dots, f_k\}$ $=$ $\{ \{v_1(x)\} \mid x \in B\} $ with
    $0 \le f_1 < f_2 < \cdots < f_k < 1$. By definition of
    $v_1 \eqv_M v_2$, the set
    $\{\{v_2(x)\} \mid x\in B \}$ is of the form
    $\{ f'_1, \dots, f'_k\}$ with
    $0 \le f'_1 < f'_2 < \cdots < f'_k < 1$, and $f_1 = 0$ iff
    $f'_1 = 0$.

    Define $f_{0}=f'_{0}=0$ and $f_{k+1} = f'_{k+1} = 1$. Consider the set
    $\{1-f_k, \dots, 1-f_2, 1-f_1\}$. Either $\{\delta_1\}$ equals some $1-f_i$, or
    $1-f_{i+1} < \{\delta_1\} < 1-f_{i}$ for some $0 \le i \le k$. 
    Choose $\{\delta_2\} = 1- f'_i$ in the former case, and
    $\{\delta_2\} = \ell$ for some $1-f'_{i+1} < \ell < 1-f'_{i}$ in
    the latter case. This choice maintains the following property for
    all clocks $x \in B$:
    \begin{align}
      \{v_1(x)\} + \{\delta_1\} < 1 & \text{ iff } \{v_2(x)\} + \{\delta_2\} < 1 
      \label{eq:5} \\
      \{v_1(x)\} + \{\delta_1\} \leq 1 & \text{ iff } \{v_2(x)\} + \{\delta_2\} \leq 1 
      \nonumber 
    \end{align}
    Moreover, observe that:
    \begin{align}
      v'_{1}(x)=v_1(x)+\delta_1 & = \Int{v_1(x)} + \Int{\delta_1} + \{v_1(x)\}
                            + \{\delta_1\} \label{eq:6}\\
      v'_{2}(x)=v_2(x)+\delta_2 & = \Int{v_2(x)} + \Int{\delta_2} +
                            \{v_2(x)\} + \{\delta_2\} \nonumber
    \end{align}
    For clocks $x \in B$ we already have $\Int{v_1(x)} = \Int{v_2(x)}$
    by definition of $v_1 \eqv_M v_2$. By construction, we have
    $\Int{\delta_1} = \Int{\delta_2}$. From (\ref{eq:5}), we
    deduce that for all $x \in B$:
    \begin{align}
      \Int{v'_1(x)} &= \Int{v'_2(x)} &\text{and}&&
      \{v'_1(x)\} = 0 &\text{ iff } \{v'_2(x)\} = 0 
      \label{eq:7} 
    \end{align}
    These arguments are sufficient to show items $1$ and $3$ in the definition of
    $\eqv_M$.  For item $2$, notice that $v'_1(x)-v'_1(y) = v_1(x)-v_1(y)$ and
    $v'_2(x)-v'_2(y) = v_2(x)-v_2(y)$.  Since $v_1, v_2$ satisfy item $2$, we conclude
    item $2$ for $v'_1$ and $v'_2$.

  \item[Reset.] Let $x$ be a history clock, and let $u_1 \in
    [x]v_1$. Therefore, $u_1(x) = 0$, and $u_1(y) = v_1(y)$ for all
    other clocks $y$. Let $u_2\in[x]v_{2}$ so that that $u_2(x) = 0$ and
    $u_2(y) = v_2(y)$ for other clocks $y$. We claim that
    $u_1 \eqv_M u_2$. Items $1,3$ of the definition follow from
    $v_1 \eqv_M v_2$, and from the fact that $u_1(x) = u_2(x) =
    0$. Item $2$ trivially holds for pairs of clocks different from
    $x$. Now, we consider differences $x - y$ and $y - x$. Notice that
    $u_1(x - y) = -v_1(y)$ and $u_2(x - y) = -v_2(y)$; and
    $u_1(y - x) = v_1(y)$ and $u_2(y - x) = v_2(y)$. Using item
    $1$ of $v_1 \eqv_M v_2$ we get $u_1\models x-y\leqlt c$ iff
    $u_2\models x-y\leqlt c$, and $u_1\models y - x \leqlt c$ iff
    $u_2\models y-x \leqlt c$.

    \item[Release.] Let $x$ be a future clock and let $u_1\in [x]v_1$. 
    By definition, $u_1$ and $v_1$ differ only in the value
    of clock $x$. Construct valuation $u_2\in [x]v_2$ as follows. 
    Set $u_{2}(x)=-\infty$ if $u_{1}(x)=-\infty$. Otherwise, set
    $\Int{u_2(x)}=\Int{u_1(x)}$; if $\{u_1(x)\} = 0$, then set
    $\{u_2(x)\} = 0$; else choose $\{u_2(x)\}$ in $(0,1)$ such that
    for all clocks $y\neq x$ with $-\infty< u_1(y) \le M$ we have:
    $\{v_1(y)\}<\{u_1(x)\}$ iff $\{v_2(y)\}<\{u_2(x)\}$, and
    $\{v_1(y)\}=\{u_1(x)\}$ iff $\{v_2(y)\}=\{u_2(x)\}$. By
    construction, $u_2$ satisfies items 1 and 3 of the definition
    of $\eqv_M$. We prove that item 2 is also satisfied.
    
    Consider a diagonal constraint $y-z\leqlt c$ with ${\leqlt}\in\{{<},{\leq}\}$ and
    $c\in\{-\infty,+\infty\}$ or $c\in\mathbb{Z}$ with $|c|\leq M$.  
    If $x\notin\{y,z\}$ then $u_{1}(y-z)=v_{1}(y-z)$ and $u_{2}(y-z)=v_{2}(y-z)$.
    Since $v_{1}\approx_{M}v_{2}$ we obtain $u_1\models y-z\leqlt c$ iff
    $u_2\models y-z\leqlt c$.  
    
    Consider the case $y\neq x=z$.  
    We have $u_1(y-x)=v_{1}(y)-u_{1}(x)$ and $u_2(y-x)=v_{2}(y)-u_{2}(x)$.
    If $u_{1}(x)=-\infty$ then $u_{2}(x)=-\infty$.
    If $v_{1}(y)\in\{-\infty,+\infty\}$ then $v_{2}(y)=v_{1}(y)$.
    In both cases, we have $u_{1}(y-x)=u_{2}(y-x)\in\{-\infty,+\infty\}$ and we get 
    $u_1\models x-y\leqlt c$ iff $u_2\models x-y\leqlt c$.
    In the remaining cases we have $u_{1}(x),u_{2}(x),v_{1}(y),v_{2}(y)\in\mathbb{R}$.
    
    If $M<v_{1}(y)$ then $M<v_{2}(y)$ and $M<u_{1}(y-x),u_{2}(y-x)<+\infty$.
    Hence, $u_1\models x-y\leqlt c$ iff $u_2\models x-y\leqlt c$. 
    Otherwise, $v_{1}(y),v_{2}(y)\leq M$ and $\Int{v_1(y)}=\Int{v_2(y)}$. 
    By construction we have $\Int{u_2(x)}=\Int{u_1(x)}$. We write
    \begin{align*}
      u_1(y-x) &= v_{1}(y)-u_{1}(x) = \Int{v_1(y)} - \Int{u_1(x)} + \{v_1(y)\} - \{u_1(x)\} \\
      u_2(y-x) &= v_{2}(y)-u_{2}(x) = \Int{v_2(y)} - \Int{u_2(x)} + \{v_2(y)\} - \{u_2(x)\} \,.
    \end{align*}
    We know that $\Int{v_1(y)}-\Int{u_1(x)}=\Int{v_2(y)}-\Int{u_2(x)}$.
    By construction, $\{v_1(y)\}<\{u_1(x)\}$ iff $\{v_2(y)\}<\{u_2(x)\}$, and
    $\{v_1(y)\}=\{u_1(x)\}$ iff $\{v_2(y)\}=\{u_2(x)\}$. We deduce once again that 
    $u_1\models x-y\leqlt c$ iff $u_2\models x-y\leqlt c$.
    
    The case $y=x\neq z$ can be proved similarly.
    \qedhere
  \end{description}
\end{proof}

\subparagraph*{\textbf{The equivalence $\sim_M$, and moving from $\sim_M$ to $\eqv_M$.}} 

The equivalence $\sim_M$ is defined on the space of all valuations.
Our goal in this part is to start from $v_1 \sim_M v_2$ and generate a valuation $v'_2$ by
modifying some values of $v_2$, so that we get $v_1 \eqv_M v'_2$.  Let us first recall the
definition of $\sim_M$, with $n$ be the number of clocks in the GTA.
\begin{itemize}
  \item First, we define $\sim_{M}$ on $\alpha,\beta\in\overline{\mathbb{R}}$ by
  $\alpha\sim_{M}\beta$ if
  $(\alpha\leqlt c \Longleftrightarrow \beta\leqlt c)$ for all ${\leqlt}\in\{<,\leq\}$ and
  $c\in\{-\infty,+\infty\}$ or $c\in\mathbb{Z}$ with $|c|\leq M$.
  In particular, $\alpha\sim_{M}\beta$ implies $\alpha=-\infty$ iff $\beta=-\infty$ and 
  $\alpha=+\infty$ iff $\beta=+\infty$. Also, if $-M\leq\alpha\leq M$ then 
  $\alpha\sim_{M}\beta$ implies $\Int{\alpha}=\Int{\beta}$ and $\{\alpha\}=0$ iff 
  $\{\beta\}=0$.

  \item For valuations $v_{1},v_{2}\in\V$ we
  define $v_{1}\sim_M v_{2}$ if $(i)$ $v_{1}(x)\sim_{nM} v_{2}(x)$ for all $x\in X$,
  and $(ii)$ $v_{1}(x) - v_{1}(y) \sim_{(n+1)M} v_{2}(x) - v_{2}(y)$ for all 
  pairs of clocks $x,y\in X$.
\end{itemize}

Notice that $\eqv_M$ and $\sim_M$ are incomparable, in the sense that
neither of them is a refinement of the other. The equivalence $\eqv_M$
constrains values up to $M$, whereas $\sim_M$ looks at values up to
$nM$, i.e., between $-nM$ and $nM$. For instance, consider
$v_1 := \langle x = M+2, y = 1\rangle$ and
$v_2 := \langle x = M + 3, y = 1 \rangle$ for some $M \ge 2$. We have
$v_1 \eqv_M v_2$, but $v_1 \not \sim_M v_2$.
For the other way around, notice that $\sim_M$ has
finite index, whereas $\eqv_M$ does not. So, $v_1 \sim_M v_2$ does not
imply $v_1 \eqv_M v_2$. To see it more closely, $v_1 \eqv_M v_2$
enforces the same integral values for all future clocks. For
clocks less than $-(n+1) M$, there is no such constraint on the actual
values in $\sim_M$.

As mentioned above, our objective is to obtain $\eqv_M$ equivalent
valuations starting from $\sim_M$ equivalent
ones. Lemma~\ref{lem:sim-M-to-region} is a first step in this
direction. It essentially shows that, when restricted to clocks within
$-M$ and $+M$, $\sim_M$ entails $\eqv_M$.

\begin{restatable}{lemma}{simMtoregion}
  \label{lem:sim-M-to-region}
  Suppose $v_1 \sim_M v_2$. Let $x, y$ be clocks such that
  $-M \le v_1(x), v_1(y) \le M$. Then, $\Int{v_1(x)} = \Int{v_2(x)}$,
  $\{v_1(x)\} = 0$ iff $\{v_2(x)\} = 0$, and
  $\{v_1(x)\} \le \{v_1(y)\}$ iff $\{v_2(x)\} \le \{v_2(y)\}$.
\end{restatable}

\begin{proof}
  Since $v_{1}\sim_{M} v_{2}$ we have in particular $v_{1}(x)\sim_{M} v_{2}(x)$ and we 
  deduce $\Int{v_1(x)} = \Int{v_2(x)}$ and $\{v_1(x)\} = 0$ iff $\{v_2(x)\} = 0$.
  For the last statement, observe that:
  \begin{align}
    v_1(x - y) & = \Int{v_1(x)} - \Int{v_1(y)} + \{v_1(x)\} -
                 \{v_1(y)\} \label{eq:8}\\
    v_2(x - y) & = \Int{v_2(x)} - \Int{v_2(y)} + \{v_2(x)\} -
                 \{v_2(y)\} \nonumber
  \end{align}
  Let $d := \Int{v_1(x)} - \Int{v_1(y)} = \Int{v_2(x)} - \Int{v_2(y)}$.
  Since $-M \le v_1(x), v_1(y) \le M$, we have $- 2M \le v_1(x - y) \le
  2M$. Therefore $-2M \le d \le 2M$. In the definition of
  $v_1 \sim_M v_2$, we have $v_1(x - y) \leqlt c$ iff
  $v_2(x - y) \leqlt c$ for all $c\in\mathbb{Z}$ with $|c| \le (n+1) M$, where
  $n\geq1$ is the number of clocks. 
  Since $-2M \le d \le 2M$, we deduce $v_1(x - y) \le d$ iff
  $v_2(x - y) \le d$. Using (\ref{eq:8}), we deduce
  $\{v_1(x)\} \le \{v_1(y)\}$ iff $v_1(x - y) \le d$ iff
  $v_2(x-y) \le d$ iff $\{v_2(x)\} \le \{v_2(y)\}$.
\end{proof}

Lemma~\ref{lem:sim-M-to-region} considers clocks within $-M$ and $+M$. What about clocks
above $M$? Directly from $v_1 \sim_M v_2$, we have $M < v_1(x)$ iff
$M < v_2(x)$, and moreover diagonal constraints up to $M$ are already
preserved by $\sim_M$. Therefore, together with
Lemma~\ref{lem:sim-M-to-region}, $v_1 \sim_M v_2$ implies
$v_1 \eqv_M v_2$ when restricted to clocks greater than $-M$. We
cannot say the same for clocks lesser than $-M$, in particular we
may have $v_1(x) = -nM-1$ and $v_2(x) = -nM - 2$. However, as shown in
the lemma below, we can choose suitable values for clocks lesser than
$-M$ to get a $\eqv_M$-equivalent valuation from a $\sim_M$-equivalent
one.

\begin{restatable}{lemma}{simtoeqv}
  \label{lem:sim-to-eqv}
  Suppose $v_1 \sim_M v_2$, and let $L = \{ x \mid -M \le
  v_1(x)\}$. There is a valuation $v'_2$ such that
  $v'_2 \da_{L} = v_2 \da_{L}$ and $v_1 \eqv_M v'_2$.
\end{restatable}

\begin{proof} We will construct a valuation $v'_2$ as required. For
  all clocks $x \in L$, the lemma imposes $v'_2(x) = v_2(x)$. The set
  $L$ includes all the history clocks. For $x\in X_{F}$ with $v_{1}(x)=-\infty$ 
  we set $v'_{2}(x)=-\infty$.
  Let $\overline{L}=\{x\mid -\infty<v_{1}(x)<-M\} \subseteq X \setminus L$ be the future
  clocks for which we still need to find a value in $v'_2$.  For all $x\in\overline{L}$,
  we set $\Int{v'_2(x)} = \Int{v_1(x)}$ and if $\{v_1(x)\}=0$, set $\{v'_2(x)\}=0$.
  It remains to pick fractional parts for clocks in $\overline{L}$ such that $0<\{v_1(x)\}<1$.

  Let $B= \{ x \mid -\infty< v_1(x) \le M \}$. 
  From Lemma~\ref{lem:sim-M-to-region}, for
  $x, y \in B \cap L$, we have $\{v_1(x)\} \le \{v_1(y)\}$ iff
  $\{v_2(x)\} \le \{v_2(y)\}$, iff $\{v'_2(x)\} \le \{v'_2(y)\}$ (as
  $v_2$ and $v'_2$ coincide on $B \cap L$). Moreover, for clocks in
  $\overline{L}$ with $\{v_1(x)\} = 0$, we have already chosen
  $\{v'_2(x)\} = 0$.  We can choose fractional values for the other clocks in
  $\overline{L}$ so that we get $\{v'_2(x)\} \le \{v'_2(y)\}$ iff
  $\{v_1(x)\} \le \{v_1(y)\}$ for all $x, y \in B$.
  We prove that $v_1 \eqv_M v'_2$.
  \begin{itemize}
    \item Let ${\leqlt}\in\{{<},{\leq}\}$ and $c\in\{-\infty,+\infty\}$ or 
    $c\in\mathbb{Z}$ with $|c|\leq M$. Since $v_1 \sim_M v_2$, we have $v_1(x)\leqlt c$ iff
    $v_2(x)\leqlt c$. By construction of $v'_2$, we have $v'_{2}(x)=v_{2}(x)$ for $x\in 
    L$ and for $x\in X\setminus L$ we have $v_1(x)\not\leqlt c$ and $v'_2(x)\not\leqlt c$.
    Therefore, $v_1(x)\leqlt c$ iff $v'_2(x)\leqlt c$ for all $x\in X$.
    
    Now, let $c\in\mathbb{Z}$ with $c<-M$.  By construction of $v'_2$, for all
    $x\in X\setminus L$ we have $v_1(x)\leqlt c$ iff $v'_2(x)\leqlt c$.
    Also, for all $x\in{L}$ we have $v_1(x)\not\leqlt c$ and $v'_2(x)\not\leqlt c$.
    Therefore, item 1 of Definition~\ref{def:eqv_M} is satisfied.

    \item  Item 3 of Definition~\ref{def:eqv_M} follows directly by construction of $v'_{2}$.

    \item Let ${\leqlt}\in\{{<},{\leq}\}$ and $c\in\{-\infty,+\infty\}$ or 
    $c\in\mathbb{Z}$ with $|c|\leq M$. Let $x,y\in X$.
     
    If $v_{1}(x)\in\{-\infty,+\infty\}$ or $v_{1}(y)\in\{-\infty,+\infty\}$ then 
    $v_{1}(x)-v_{1}(y)=v'_{2}(x)-v'_{2}(y)\in\{-\infty,+\infty\}$ and we deduce that
    $v_{1}\models x-y\leqlt c$ iff $v'_{2}\models x-y\leqlt c$. We assume below that 
    $v_{1}(x),v_{1}(y)\in\mathbb{R}$. We have
    \begin{align*}
      v_1(x - y) & = \Int{v_1(x)} - \Int{v_1(y)} + \{v_1(x)\} - \{v_1(y)\} \\
      v'_2(x - y) & = \Int{v'_2(x)} - \Int{v'_2(y)} + \{v'_2(x)\} - \{v'_2(y)\}
    \end{align*}
    
    Since $v_1 \sim_M v_2$ and $v'_2 \da_{L} = v_2 \da_{L}$, we get immediately
    $v_{1}\models x-y\leqlt c$ iff $v'_{2}\models x-y\leqlt c$ for all $x,y\in L$.
    When both $x, y \in B$, we get this condition from items $1,3$ of $\eqv_M$. 
    
    The two remaining cases is when $-\infty<v_1(x)<-M$ and $M<v_1(y)<+\infty$, 
    or when $-\infty<v_1(y)<-M$ and $M<v_1(x)<+\infty$.
    In the first case, $-\infty<v'_2(x)<-M$ by construction and 
    $M<v_2(y)=v'_{2}(y)<+\infty$ since $v_{1}\sim_{M} v_{2}$.
    Therefore, $-\infty<v_1(x)-v_{1}(y)<-2M$ and $-\infty<v'_2(x)-v'_{2}(y)<-2M$.
    We deduce that $v_{1}\models x-y\leqlt c$ iff $v'_{2}\models x-y\leqlt c$.
    The proof is similar in the second case.
    \qedhere
  \end{itemize}
\end{proof}

Finally, we show that if we have a run between $\sim_M$ equivalent
valuations, we can extract a run between $\eqv_M$ equivalent
valuations, simply by changing the last released values of future
clocks.  Suppose there is a run $\rho$ from configuration $(q, v_1)$
to configuration $(q, v_k)$ such that $v_1 \sim_M v_k$, and all future
clocks are released in $\rho$.  By Lemma~\ref{lem:sim-to-eqv}, there
is a $v'_k$ satisfying $v_1 \eqv_M v'_k$ that differs from $v_k$ only in the clocks that
are less than $-M$.  In order to reach $v'_k$ from $v_1$
using the same sequence of transitions as in $\rho$, it is enough to
choose a suitable shifted value during the last release of the clocks
that were modified.  This gives a new run $\rho'$.  The non-trivial
part is to show that $\rho'$ is indeed a run, that is: all guards are
satisfied by the new values.  We depict this situation in
Figure~\ref{fig:diagonal-after-changing-vk}.  The modified clocks are
those that are less than $-M$ in $v_k$.  Clock $x$ is one such. The black dot
represents its value in $v_k$, and the blue dot is its value in
$v'_k$. Its new value is still $<-M$. In the run $\rho'$, clock $x$ is
released to a suitably shifted value at its last release point. Notice
that from this last release point till $k$, clock $x$ stays below $-M$
in both $\rho$ and $\rho'$. Therefore, all non-diagonal constraints
$x \leqlt c$ that were originally satisfied in $\rho$ continue to get
satisfied in $\rho'$. Showing that all diagonal constraints are still
satisfied is not as easy. Here, we make use of the safety assumption.
Let us look at a diagonal constraint $x - y$, and a situation as in
Figure~\ref{fig:diagonal-after-changing-vk} where the last release of
$y$ happens after the last release of $x$. For simplicity, let us
assume there is no release of $y$ in between these two points.

\begin{figure}[tbh]
  \centering
  \begin{tikzpicture}
    \node at (1,0) {\footnotesize $v_1$}; \node at (8.2,0)
    {\footnotesize $v_k$, \textcolor{blue}{$v'_k$}};
    \draw [thick, gray] (1.3,0) to (7.7,0); \draw [thin] (3, 0.1) --
    (3, -0.1); \draw [thin] (6, 0.1) -- (6, -0.1); \node at (3, 0.4)
    {\small last release of $x$}; \node at (6, 0.4) {\small last release of $y$};

    \draw [thin, gray] (4,-0.5) -- (4, -2.3); \draw (3.98, -0.5) --
    (4.02, -0.5); \draw (3.98, -1) -- (4.02, -1); \node [left] at
    (3.98,-0.5) {\small $0$}; \node [left] at (3.98,-1) {\small $-M$};
    \fill (4, -0.8) circle (0.8pt); \fill (4, -1.8) circle (0.8pt);
    \node [right] at (4,-0.8) {\small $y$}; \node [right] at (4,-1.8)
    {\small $x$}; \fill [blue] (4, -2.1) circle (0.8pt); \node [right]
    at (4, -2.1) {\small \textcolor{blue}{$x$}};

    \begin{scope}[xshift=2cm]
      \draw [thin, gray] (4,-0.5) -- (4, -2.3); \draw (3.98, -0.5) --
      (4.02, -0.5); \draw (3.98, -1) -- (4.02, -1); \node [left] at
      (3.98,-0.5) {\small $0$}; \node [left] at (3.98,-1) {\small $-M$};
      \fill (4, -0.5) circle (0.8pt); \fill (4, -1.5) circle (0.8pt);
      \node [right] at (4,-0.5) {\small $y$}; \node [right] at (4,-1.5)
      {\small $x$};

      \fill [blue] (4, -1.8) circle (0.8pt); \node [right] at (4,
      -1.8) {\small \textcolor{blue}{$x$}};
    \end{scope}

    \begin{scope}[xshift=4cm]
      \draw [thin, gray] (4,-0.5) -- (4, -2.3); \draw (3.98, -0.5) --
      (4.02, -0.5); \draw (3.98, -1) -- (4.02, -1); \node [left] at
      (3.98,-0.5) {\small $0$}; \node [left] at (3.98,-1) {\small $-M$};
      \fill (4, -1.2) circle (0.8pt);
      \node [right] at (4,-1.2) {\small $x$};
      
      \fill [blue] (4, -1.5) circle (0.8pt); \node [right] at (4, -1.5)
      {\small \textcolor{blue}{$x$}};
    \end{scope}
  \end{tikzpicture}
  \caption{An illustration for the proof of Lemma~\ref{lem:sim-to-eqv-2}}
  \label{fig:diagonal-after-changing-vk}
\end{figure}
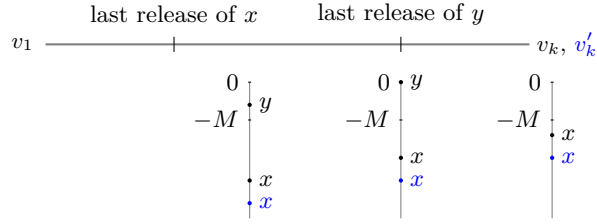

We divide the run into three parts: the left part is the one before
the last release of $x$, the middle part is the one between the two
release points, and the right part is the rest of the run, to the
right of the release of $y$. In the left part, the values of $x$ and
$y$ are the same in both $\rho$ and $\rho'$, and so the diagonal
constraints continue to get satisfied. In the right part, the value of
$x - y$ equals $v'_k(x) - v'_k(y)$. Using $v'_k \eqv_M v_1$ and
$v_1 \sim_M v_k$, we can argue that $v'_k$ and $v_k$ satisfy the same
diagonal constraints up to constant $M$. This takes care of the right
part. The middle part is the trickiest. In this part, we know that $x$
remains less than $-M$ in both $\rho$ and $\rho'$. The value of $y$ is
the same in both $\rho$ and $\rho'$. But what about the difference
$x - y$? Can it be, say $-1$ in $\rho$ and $-2$ in $\rho'$?  Here is
where we use the safety assumption to infer the value of $x -
y$. Before $y$ is released, its value should be $0$. At that point,
$x$ is still less than $-M$ (in both the runs). Therefore $x - y < -M$
just before $y$ is last released. As the differences do not change, we
see that $x - y < - M$ in the middle part, for both runs. Hence the
diagonal constraints continue to hold in $\rho'$. We formalize these
observations in Lemma~\ref{lem:sim-to-eqv-2}, where we exhaustively
argue about all the different cases.

\begin{restatable}{lemma}{simtoeqvb}
  \label{lem:sim-to-eqv-2}
  Consider a safe GTA $\Aa$. Let $\rho: (q_1, v_1)
  \xra{\delta_1, t_1} (q_2, v_2) \xra{\delta_2, t_2} \cdots (q_k, v_k)$ 
  be a run of $\Aa$ such that $v_1 \sim_M v_k$ and for every future clock
  $x$, either $x$ is released in the transition sequence $t_1 \dots t_{k-1}$ or 
  $v_1(x) = -\infty$.  Let
  $L = \{x \mid -M \le v_1(x)\}$. Let $v'_k$ be a valuation such that
  $v'_k \da_{L} = v_k \da_L$ and $v_1 \eqv_M v'_k$.
  Then, there exists a run of the form
  $\rho': (q_1, v_1) = (q_1, v'_1) \xra{\delta_1, t_1} (q_2, v'_2)
  \xra{\delta_2, t_2} \cdots (q_k, v'_k)$ in $\Aa$, leading to
  $(q_k, v'_k)$ from $(q_1,v_1)$.
\end{restatable}

\begin{proof}
  First, the existence of $v'_k$ comes from Lemma~\ref{lem:sim-to-eqv}.
  Notice that $\rho'$ and $\rho$ have the same time elapse values.
  Let $\overline{L} = \{ x \mid -\infty < v_{1}(x),v_{k}(x) < -M \}$.
  These are the clocks where $v'_k$ may differ from
  $v_k$.  To make $\rho'$ end with $v'_k$, it is enough to correctly
  choose the last released values for clocks in $\overline{L}$ along
  the run. 

  Without loss of generality, assume that each $t_i$ is an atomic
  program, consisting of either a guard or a change. If the transition
  is a ``change'', we can assume the guard to be $\top$. For each future clock $x\in 
  X_{F}$ which is released in the transition sequence $t_{1}\cdots t_{k-1}$,
  let $j_x$ be the last index such that $x$ is
  released in $t_{j_x}$. For each clock $x\in\overline{L}$,
  we will now give a different released value at $j_x$: set
  $v'_{j_x+1}(x) = v'_k(x) - \sum_{i=j_x+1}^{m-1} \delta_i$. 
  This construction fixes the values of all clocks along
  the run $\rho'$. However, we need to show that the new valuations
  still satisfy the guards along the way. 
  We will prove that $v_i$ satisfies an atomic constraint iff $v'_i$ does. Let
  ${\leqlt}\in\{{<},{\leq}\}$ and 
  $c \in \{-\infty, +\infty\} \cup \{ d\in\mathbb{Z} \mid |d| \le M\}$.

  Consider a non-diagonal constraint $x \leqlt c$ in $t_i$. If
  $v'_i(x) = v_i(x)$, we are done. Otherwise, $x \in \overline{L}$,
  $-\infty < v_k(x), v'_k(x) < -M$ and $j_x < i$.
  As $x$ is not released in $t_i,\ldots,t_{k-1}$, we get
  $-\infty < v_i(x)\leq v_k(x) < -M$ and $-\infty < v'_i(x)\leq v'_k(x) < -M$.
  This shows $v_i\models x\leqlt c$ iff $v'_i\models x\leqlt c$.

  Consider a diagonal constraint of the form $x-y \leqlt c$. 
  If $v_{i}(x)\in\{-\infty,+\infty\}$ then $v'_{i}(x)=v_{i}(x)$ and 
  if $v_{i}(y)\in\{-\infty,+\infty\}$ then $v'_{i}(y)=v_{i}(y)$. Hence, if
  $v_{i}(x)\in\{-\infty,+\infty\}$ or $v_{i}(y)\in\{-\infty,+\infty\}$) then 
  $v_{i}(x)-v_{i}(y)=v'_{i}(x)-v'_{i}(y)\in\{-\infty,+\infty\}$. We deduce that in this 
  case $v_{i}\models x-y\leqlt c$ iff $v'_{i}\models x-y\leqlt c$. 
  We assume below that $v_{i}(x),v_{i}(y),v'_{i}(x),v'_{i}(y)\in\mathbb{R}$. 
  
  If $v_{i}(x)=v'_{i}(x)$ and $v_{i}(y)=v'_{i}(y)$ then clearly
  $v_{i}\models x-y\leqlt c$ iff $v'_{i}\models x-y\leqlt c$.
  Otherwise, either $x\in\overline{L}$ and $j_{x}<i$, or 
  $y\in\overline{L}$ and $j_{y}<i$, or both.
  \begin{itemize}
    \item Suppose $x,y\in X_{F}$ and $j_{x},j_{y}<i$ 
    (right part of Figure~\ref{fig:diagonal-after-changing-vk}). Then,
    $v_i(x-y) = v_k(x-y)$ and $v'_i(x-y) = v'_k(x-y)$. 
    Observe that $v_k\models x-y\leqlt c$ iff $v_1\models x-y\leqlt c$
    (since $v_1 \sim_M v_k$), iff $v'_k\models x-y\leqlt c$ (since
    $v_1 \eqv_M v'_k$).  Therefore, $v_i\models x-y\leqlt c$ iff
    $v'_i\models x-y\leqlt c$.    
  
    \item Suppose $x\in\overline{L}$, $j_{x}<i$ and $y\in X_{H}$. 
    Then, $-\infty<v_{i}(x)\leq v_{k}(x)<-M$, and
    $-\infty<v'_{i}(x)\leq v'_{k}(x)<-M$, and 
    $0\leq v_{i}(y)=v'_{i}(y)<+\infty$. We deduce that $v_i\models x-y\leqlt c$ iff
    $v'_i\models x-y\leqlt c$.
    
    \item The case $y\in\overline{L}$ with $j_{y}<i$ and $x\in X_{H}$ is proved similarly.
    
    \item Suppose $x,y\in X_{F}$ and $j_{x}<i\leq j_{y}$.
    Let $i\leq j\leq k-1$ be the first release of $y$ in the sequence $t_{i}\cdots t_{k-1}$. 
    Notice that $v_i(x-y) = v_j(x-y)$ and $v'_i(x-y) = v'_j(x-y)$. 
    So, it is enough to show $v_j\models x-y \leqlt c$ iff $v'_j\models x-y \leqlt c$. 
    Since $\Aa$ is a safe GTA we have $v_j(y)\in\{-\infty,0\}$. 
    Since $v_{i}(y)\neq-\infty$ we get $v_{j}(y)=0$. 
    Since $j\leq j_{y}$ we have $v'_j(y)=v_j(y)=0$.
    Therefore, $v_j\models x-y \leqlt c$ iff $v_j\models x \leqlt c$ iff
    $v'_j\models x \leqlt c$ (non-diagonal case above) iff $v'_j\models x-y \leqlt c$.
  
    \item The case $x,y\in X_{F}$ with $j_{y}<i\leq j_{x}$ is proved similarly.  
    \qedhere
  \end{itemize}
\end{proof}

We lift this argument to the level of zones, to obtain one of the
main results of this paper.
Before showing the result in Lemma~\ref{prop:gta-liveness}, we first prove an intermediate result.

\begin{lemma}\label{lem:all-infty}
  Let $Z$ be a non-empty zone and let $R\subseteq X_{F}$ be a set of 
  future clocks $x$ such that $v_{x}(x)=-\infty$ for some valuation $v_{x}\in Z$.
  There are valuations $v\in Z$ such that $v(x)=-\infty$ for all $x\in R$.
\end{lemma}

\begin{proof}
  Let $\GG$ be the canonical distance graph for the non-empty zone $Z$.
  Since for each $x\in R$ there is a valuation $v_{x}$ in $Z$ such that 
  $v_{x}(x)=-\infty$ we deduce that for each $x\in R$ we have $\GG_{x0}=(\leq,+\infty)$.
  Let $\GG'$ be the distance graph obtained from $\GG$ by setting 
  $\GG'_{0x}=(\leq,-\infty)$ for each $x\in R$.
  The distance graph $\GG'$ is in standard form.
  Moreover, $\GG'$ has no negative cycles: a \emph{simple} cycle in $\GG'$ which is not 
  in $\GG$ is of the form $0\to x\to 0$ with $x\in R$ and its weight is 
  $(\leq,-\infty)+(\leq+\infty)=(\leq,+\infty)$.
  Therefore, $\sem{\GG'}$ is non-empty.
  Any valuation $v\in\sem{\GG'}\subseteq\sem{\GG}=Z$ is such that $v(x)=-\infty$ for all 
  $x\in R$.
\end{proof}

\begin{figure}[tbh]
  \centering
  \scalebox{.9}{
  \begin{tikzpicture}[state/.style={circle, inner sep=2pt, minimum size = 3pt, draw, thick}]
  \begin{scope}[]
    \node (0) at (4,0) {\small $(q,w_1)$}; 
    \node (1) at (6,0) {\small $(q,u_0)$};
    \node [label={[rotate=90, anchor=east]below:$\in$}] (3)  at (4.2,1) { $Z$};
    \node [label={[rotate=90, anchor=east]below:$\in$}] (3)  at (6.2,1) { $Z'$};

    \begin{scope}[thick, ->,>=stealth, auto]
      \draw (0) to node {$\sigma$} (1); 
    \end{scope}
  \end{scope}

  \begin{scope}[dashed,color=red,-,>=stealth, auto]
    \draw (4,0) to node {$\preceq$} (4,-2); 
    \draw (6,0) to node {$\preceq$} (6,-2);
  \end{scope}

  \begin{scope}[yshift=-2cm]
    \node (3) at (6,0) {\small $(q,u'_0)$}; 
    \node (2) at (4,0) {\small $(q,w'_1)$};
    \node [label={[rotate=90, anchor=west]above:{\scriptsize$\ni$}}] (6)  at (4.2,-0.8) { \small $Z'$};

    \begin{scope}[thick, ->,>=stealth, auto]
      \draw (2) to node {$\sigma$} (3); 
    \end{scope}
  \end{scope}

  \begin{scope}[yshift=-5cm]
    \node (3) at (6,0) {\small $(q,u'_0)$}; 
    \node (2) at (4,0) {\small $(q,w'_1)$};
    \node (1) at (2,0) {\small $(q,w_2)$};
    \node [label={[rotate=90, anchor=east]below:{\scriptsize$\in$}}] (6)  at (4.2,0.8) { \small $Z'$};
    \node [label={[rotate=90, anchor=east]below:{\scriptsize$\in$}}] (8)  at (2.2,0.8) { \small $Z$};

    \begin{scope}[thick, ->,>=stealth, auto]
      \draw (2) to node {$\sigma$} (3); 
      \draw (1) to node {$\sigma$} (2);
    \end{scope}
  \end{scope}

  \begin{scope}[dashed,color=red,-,>=stealth, auto]
    \draw (4,-5) to node {$\preceq$} (4,-7); 
    \draw (6,-5) to node {$\preceq$} (6,-7); 
    \draw (2,-5) to node {$\preceq$} (2,-7); 
  \end{scope}

  \begin{scope}[yshift=-7cm]
    \node (3) at (6,0) {\small $(q,u''_0)$}; 
    \node (2) at (4,0) {\small $(q,w''_1)$};
    \node (1) at (2,0) {\small $(q,w'_2)$};
    \node [label={[rotate=90, anchor=west]above:{\scriptsize$\ni$}}] (6)  at (2.2,-0.8) { \small $Z'$};
    
    \begin{scope}[thick, ->,>=stealth, auto]
      \draw (2) to node {$\sigma$} (3); 
      \draw (1) to node {$\sigma$} (2);
    \end{scope}
  \end{scope}

  \begin{scope}[yshift=-10cm]
    \node (3) at (6,0) {\small $(q,u'''_0)$}; 
    \node (1) at (2,0) {\small $(q,w'_2)$};
    \node (0) at (0,0) {\small $(q,w_3)$};
    \node [label={[rotate=90, anchor=east]below:{\scriptsize$\in$}}] (6)  at (2.2,0.8) { \small $Z'$};
    \node [label={[rotate=90, anchor=east]below:{\scriptsize$\in$}}] (8)  at (0.2,0.8) { \small $Z$};

    \begin{scope}[thick, ->,>=stealth, auto]
      \draw (1) to node {$\sigma^2$} (3); 
      \draw (0) to node {$\sigma$} (1);
    \end{scope}
  \end{scope}

  \begin{scope}[dashed,color=red,-,>=stealth, auto]
    \draw (6,-10) to node {$\preceq$} (6,-12); 
    \draw (0,-10) to node {$\preceq$} (0,-12); 
  \end{scope}

  \begin{scope}[yshift=-12cm]
    \node (3) at (6,0) {\small $(q,u'''_0)$}; 
    \node (0) at (0,0) {\small $(q,w'_3)$};
    \node [label={[rotate=90, anchor=west]above:{\scriptsize$\ni$}}] (6)  at (0.2,-0.8) { \small $Z'$};
    
    \begin{scope}[thick, ->,>=stealth, auto]
      \draw (0) to node {$\sigma^3$} (3); 
    \end{scope}		
  \end{scope}

  \node (a) at (-2,-11) {$Z \preceq_{\G(q)} Z'$}; 
  \node (a) at (-2,-6) {$Z \preceq_{\G(q)} Z'$}; 
  \node (a) at (-2,-1) {$Z \preceq_{\G(q)} Z'$}; 
  \node (b) at (-2,-3.5) {Post-property}; 
  \node (b) at (-2,-8.5) {Post-property}; 
\end{tikzpicture}
}
\caption{Construction of the run iterating the sequence $\s = t_1 \dots t_{k-1}$.}
\label{fig:simulation-runs}
\end{figure}
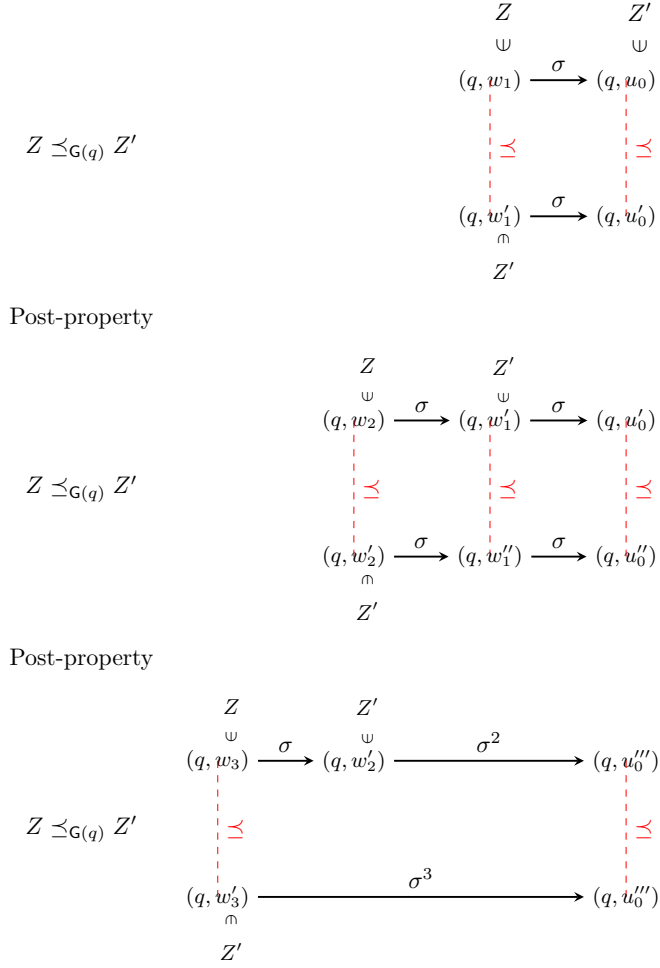

\begin{restatable}{theorem}{propgtaliveness}
  \label{prop:gta-liveness}
  Let $(q, Z) = (q_1, Z_1) \xra{t_1} (q_2, Z_2) \xra{t_2} \cdots
  \xra{t_{k-1}} (q_k, Z_k) = (q, Z')$ be a run in the zone graph such
  that $(q,Z)\preceq(q,Z')$, $(q,Z')\preceq(q,Z)$ and for every future clock $x$, 
  either $x$ is released in the sequence $t_1 \dots t_{k-1}$, 
  or there is a valuation $v_{x}\in Z'$ with $v_{x}(x)=-\infty$.
  Then, there is a valuation $v \in Z$ and an infinite run starting from
  $(q,v)$ over the sequence of transitions $(t_1 \dots t_{k-1})^\omega$.
\end{restatable}

\begin{proof}
  Let $R\subseteq X_{F}$ be the set of future clocks which are not released in the 
  transition sequence $t_{1}\cdots t_{k-1}$. 
  Using Lemma~\ref{lem:all-infty}, we can pick a valuation $u_0 \in Z'$ with
  $u_{0}(x)=-\infty$ for all $x\in R$.
  By the post-property of the zone graph, there exists $w_1 \in Z$ such that
  $(q, w_1) \xrightarrow{t_1 \dots t_{k-1}} (q, u_0)$. 
  Note that $w_{1}(x)=-\infty$ for all $x\in R$.
  Since $(q,Z)\preceq(q,Z')$, there exists $w'_1\in Z'$ such that $(q,w_1)\preceq(q,w'_1)$.
  From the third condition of a simulation in Definition~\ref{def:simulation},
  we get $w'_{1}(x)=-\infty$ for all $x\in R$.
  Therefore, there is a run $\rho_1:= (q,w'_1) \xrightarrow{t_1 \dots t_{k-1}} (q,u'_0)$ 
  such that $(q,u_0)\preceq(q,u'_0)$. 
  
  Set $u_1 = w'_1$ and repeat the argument above.  There exists $w'_2 \in Z'$ from which
  there exists a run $\rho_2:= (q, w'_2) \xrightarrow{t_1 \dots t_{k-1}} (q, w''_1)$ with
  $(q,w'_1)\preceq(q,w''_1)$.  Using this simulation, $\rho_2$ can be extended
  with a run similar to $\rho_1$ to give:
  $(q,w'_2) \xra{(t_1 \dots t_{k-1}) (t_1 \dots t_{k-1})} (q,u''_0)$ with
  $(q,u_{0})\preceq(q,u''_{0})$. 
  
  Repeated application of this argument shows for each $\ell\geq1$ the existence
  of a valuation $v_\ell \in Z'$  from which
  $(t_1 \dots t_{k-1})$ can be iterated $\ell$ times:
  $(q, v_\ell) \xra{(t_1 \dots t_{k-1})} (q, v_{\ell-1}) \cdots
  (q, v_1) \xra{t_1 \dots t_{k-1}} (q, v_0)$ with $v_{\ell}(x)=-\infty$ for all $x\in R$.
  This is illustrated in Figure~\ref{fig:simulation-runs}.

  Pick a large $\ell$ which is strictly greater than the number of
  equivalence classes of $\sim_M$. Therefore there exist $i>j$ 
  such that $v_i \sim_M v_j$, and there is a run
  $(q, v_i) \xra{(t_1 \dots t_{k-1})^{i-j}} (q, v_j)$ with $v_{i}(x)=-\infty$ for all 
  $x\in R$. By Lemma~\ref{lem:sim-to-eqv-2}, there is a run
  $(q, v_i) \xra{(t_1 \dots t_{k-1})^{i-j}} (q, v'_j)$ such that
  $v_i \eqv_M v'_j$. Since $\eqv_M$ is a bisimulation
  (Lemma~\ref{lem:eqv-is-bisimulation}), this implies there is an
  infinite run where the segment $\xra{(t_1 \dots t_{k-1})^{i-j}}$ can
  be iterated infinitely often. In particular, this gives a run of the
  required form.
\end{proof}

Finally, combining Theorem~\ref{prop:gta-liveness} and Lemma~\ref{lem:non-zeno}, we get an
algorithm for liveness: we construct the zone graph with simulation equivalence and check
for a reachable cycle that contains an accepting state and where every future clock $x$
which is not released during the cycle may take value $-\infty$ in some valuations of the
zones in the cycle.

\section{Translating MITL to GTA}
\label{sec:mitl-gta}

We first introduce the preliminaries for Metric Interval Temporal
Logic. Let $\Prop$ be a finite nonempty set of atomic propositions.
The alphabet $\Sigma$ that we consider is the set of subsets of
$\Prop$.
The set of $\mitl$ formulae over the set of atomic propositions
$\Prop$ is defined as
$$\varphi := p ~|~ \varphi \wedge \varphi  ~|~  \neg \varphi ~|~  \Next_{I} \varphi  ~|~ \varphi \Until_{I} \varphi$$ 
where $p \in \Prop$, and $I$ is either $[0,0]$, or a non-singleton
(open, or closed) interval whose end-points come from
$\Nat \cup \{\infty\}$.  In other words, if the end-points of the
interval are $a$ and $b$ respectively, then either $a=b=0$, or
$a,b \in \Nat \cup \{\infty\}$ and $a < b$.

We will now define the \emph{pointwise semantics} of MITL formulae
inductively as follows.  A timed word
$w = (a_0,\tau_0)(a_1,\tau_1)(a_2,\tau_2) \cdots $ is said to satisfy
the MITL formula $\varphi$ at position $i\geq0$, denoted as
$(w,i) \models \varphi$ if (omitting the classical boolean
connectives)
\begin{itemize}
\item $(w,i) \models p$ if $p \in a_i$
  
\item $(w,i) \models \Next_{I} \varphi$ if $(w,i+1) \models \varphi$
  and $\tau_{i+1} - \tau_{i} \in I$.
\item $(w,i) \models \varphi_1 \Until_{I} \varphi_2$ if there exists
  $j\geq i$ s.t. $(w,j) \models \varphi_2$, $(w,k) \models \varphi_1$
  for all $i \leq k < j$, and $\tau_{j} - \tau_{i} \in I$.
\end{itemize}

For a formula $\varphi$ we can also define its pointwise semantics as a function from (non-Zeno) timed words to sequences of Booleans.
Let $\TB$ denote $(\{0,1\}\times \mathbb{R})^\omega$ and $\TtwoB$ denote $(\{0,1\}^2\times \mathbb{R})^\omega$, respectively. 
Then for $\varphi\in\mitl$, we define $\sem{\varphi}\colon\TW\rightarrow \TB$ as a total function, i.e., $\dom{\sem{\varphi}}=\TW$ and for any $w=(a_0,\tau_0)(a_1,\tau_1)\ldots$,  we have $\sem{\varphi}(w)=(b_0,\tau_0)(b_1,\tau_1)\ldots$, where for all $i$, $b_i=1$ if $(w,i)\models \varphi$ and $b_i=0$ otherwise.
Finally, given an $\mitl$ formula $\varphi$ we define its language $\lang(\varphi)=\{w\in 
\TW\mid (w,0) \models \varphi\}$.

Our goal is to construct a GTA with outputs for an MITL formula
$\varphi$, which reads the timed word and outputs $1$ at position $i$
iff $(w, i) \models \varphi$. More precisely, there is a unique run of
the GTA on $w$:
$(q_0, v_0) \xra{\delta_0, t_0} (q_1, v_1) \xra{\delta_1, t_1}
\cdots$, where the output of each transition $t_i$ equals $1$ iff
$(w, i) \models \varphi$. We refer to GTA with outputs as Generalized
Timed Transducers (GTT) and we will formally define them next. 

\subsection{Generalized Timed Transducers (GTT)}

\begin{definition}[\GTTFULL]\label{defn:gtt}
  A \emph{\GTTfull} (\GTT) is a tuple $\Tt = (Q, \Sigma, \Gamma, X, \Delta, \mu, \Ii,
  (Q_{f}, g_f))$, where $\Sigma$ and $\Gamma$ are the finite \emph{input} and \emph{output
  alphabets}, $\Aa_{\Tt} = (Q, \Sigma, X, \Delta, \Ii, (Q_{f}, g_f))$ is a \GTA, called
  the \emph{underlying \GTAfull} of the \GTT\ $\Tt$, and $\mu \colon \Delta \to \Gamma^*$ is
  the \emph{output function}.
\end{definition}

\begin{figure}[htb]
	\centering
  \includegraphics[page=10,scale=1]{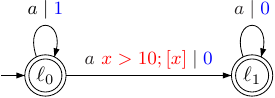}
  \hfill
  \raisebox{7mm}{$\begin{array}{rcccccccc}
    w: & a & a & a & a & a & a & a & \cdots  \\
    \tau: & 1.3 & 5.7 & 7.8 & 10 & 11 & 21.2 & 100.7 & \cdots  \\
    \sem{\Tt}(w): & 1 & 1 & 1 & 1 & 1 & 1 & 1 & \cdots  \\
    \tau: & 1.3 & 5.7 & 7.8 & 10 & 11 & 21.2 & 100.7 & \cdots  
  \end{array}$}
	\caption{An \GTTfull\ $\Tt$ on left, and a run of $\Tt$ on a timed word $w$ with the output on right.}
	\label{fig:gtt-eg}
\end{figure}

Note that associated to each transition of a \GTT\ , we have a tuple consisting of an input letter, an instantaneous timed program, and an output letter.
The semantics of the \GTT\ $\Tt$ defined above is given by a transition system $\TS_{\Tt}$ whose states are \emph{configurations} $(q,v)$ of $\Tt$, where $q \in Q$ and $v\in\V$ is a valuation. 

An example of a \GTT\ and its run is given in Figure~\ref{fig:gtt-eg}. An accepting run $\rho$ on the input timed word $w=(a_{0},\tau_{0})(a_{1},\tau_{1})\cdots$
and with sequence of discrete transitions $t_{0}t_{1}\cdots$ produces the output timed
word $\rho(w)=(\mu(t_{0}),\tau_{0})(\mu(t_{1}),\tau_{1})\cdots$.  
Note that for a \GTTfull\ $\tra$ and an input timed word $w$, there could be several
accepting runs of $\Tt$ over $w$. Hence, the semantics of $\Tt$
is a relation $\sem{\Tt} \subseteq \TWi \times \TWo$ defined by
$\sem{\Tt}=\{(w,\rho(w))\mid \rho \text{ is an accepting run on } w\}$.
The \emph{domain} of $\Tt$ is the set $\dom{\Tt}$ of timed words $w$ such that $\Tt$ has
an accepting run over $w$.  
We say that $\Tt$ is functional if the relation $\sem{\Tt}$ is a partial function.
In which case, we write $\sem{\Tt}\colon\TWi\to\TWo$ and for $w\in\dom{\Tt}$, we write 
$\sem{\Tt}(w)$ for the output word associated with $w$ by $\Tt$.

A \GTTfull\ $\tra$ is called \emph{unambiguous} if every input word $w\in\TWi$ admits at
most one accepting run.  Clearly, an unambiguous \GTT\ is functional, but the converse
need not be true.  A deterministic \GTT\ is unambiguous, hence functional.  There is also
a dual notion, generalizing prophetic (untimed) automata from~\cite{CartonMichel03}, and
corresponding to co-determinism for infinite words, which also implies unambiguity, hence
functionality.

\subsection{The compositional approach}

In this section, we explain the overall approach of our translation that goes via transducers, {\em assuming} that for each timed operator we can build a \GTT. 


\begin{wrapfigure}{r}{3.2cm}
  \centering
  \scalebox{1}{
    \begin{tikzpicture}[level distance=0.75cm,
      level 1/.style={sibling distance=1.5cm},
      level 2/.style={sibling distance=1.5cm}]
      \node {$\U_{(1,2)}$}
        child {node {$\neg$}
          child {node {$p$}}       
        }       
        child {node {$\X_{(5,7)}$}
          child {node {$\land$}
            child {node {$q$}}
            child {node {$r$}}
            }
          };
    \end{tikzpicture}
  }
  \caption{\small Parse tree for the MILT formula \\
  $\neg p \U_{(1,2)} X_{(5,7)} (q \land r) $}
  \label{fig:parse-tree}
\end{wrapfigure}
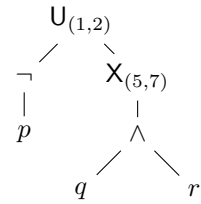
At a high level, we will convert an $\mitl$ formula to a generalized timed transducer (\GTT) which uses history and future clocks~\cite{AkshayGGJS23}, from which we obtain a \GTA. 
Our construction can be viewed as 
structural induction on the \emph{parse tree} of the $\mitl$ formula, where we show how to build a \GTT\ for atomic propositions, and then for each Boolean and temporal operator, and finally 
we compose these \GTT\ bottom up to obtain the \GTT\ for each subformula, which by structural induction finally gives us the \GTT\ for the full formula.

Consider the the $\mitl$ formula $\varphi \equiv \neg p ~\U_{(1,2)}~ \X_{(5,7)} (q \land r)$, and its parse tree depicted in Figure~\ref{fig:parse-tree}. 
Here, the subformulae of $\varphi$ are 
$\psi_0 \equiv \varphi$, 
$\psi_1 \equiv \neg p$, 
$\psi_2 \equiv \X_{(5,7)} (q \land r)$,
$\psi_3 \equiv p$, 
$\psi_4 \equiv q \land r$,
$\psi_5 \equiv q$,
$\psi_6 \equiv r$. 
Observe that each node of the parse tree  is either an atomic propositions (at the leaves), or a boolean or temporal operator (at the non-leaf nodes).
Further, each subformula of $\varphi$ corresponds to a node of the parse tree of $\varphi$ in the sense that the subtree rooted at that node is the parse tree of the subformula.

We now discuss how the transducers for the subformulae can be combined to get the transducer for a given formula.  
This is done by considering two standard constructions for transducers, namely composition and product, which we formally define next.

\begin{definition}~\label{defn:composition}  
  Let $\Tt_1 = (Q_1, A, B, X_1, \Delta_1, \mu_1, \Ii_1, (Q^1_{f}, g^1_f))$ be a \GTTfull\ from the alphabet $A$ to $B$ and
  $\Tt_2 = (Q_2, B, C, X_2, \Delta_2, \mu_2, \Ii_2, (Q^2_{f}, g^2_f))$ be a \GTTfull\ from the alphabet $B$ to $C$.
  Then, the composition of $\Tt_1$ and $\Tt_2$ is a \GTTfull\ from the alphabet $A$ to $C$ given by $\Tt_2 \circ \Tt_1 = (Q_1 \times Q_2, A, C, X_1 \cup X_2, \Delta, \mu, \Ii_1 \times \Ii_2 , ( (Q^1_{f} \times Q_2, g^1_f) \cup (Q_1 \times Q^2_{f}, g^2_f)))$
  where the transition function $\Delta$ and the output function $\mu$ are defined as
  \begin{prooftree}
    \def\defaultHypSeparation{\hskip .05in}
    \AxiomC{$p \xrightarrow{a\mid b} q$ in $\Tt_1$,}
    \AxiomC{$p' \xrightarrow{b\mid c} q'$ in $\Tt_2$}
    \RightLabel{\scriptsize(1)}
    \BinaryInfC{$(p,p') \xrightarrow{a\mid c} (q,q')$  in $\Tt_2 \circ \Tt_1$}
  \end{prooftree}
\end{definition}

\begin{definition}~\label{defn:product}  
  Let $\Tt_1 = (Q_1, A, B, X_1, \Delta_1, \mu_1, \Ii_1, (Q^1_{f}, g^1_f))$ be a \GTTfull\ from the alphabet $A$ to $B$ and
  $\Tt_2 = (Q_2, B, C, X_2, \Delta_2, \mu_2, \Ii_2, (Q^2_{f}, g^2_f))$ be a \GTTfull\ from the alphabet $B$ to $C$.
  Then, the product of $\Tt_1$ and $\Tt_2$ is a \GTTfull\ from the alphabet $A$ to $B \times C$ given by $\Tt_1 \times \Tt_2 = (Q_1 \times Q_2, A, B \times C, X_1 \cup X_2, \Delta, \mu, \Ii_1 \times \Ii_2, ( (Q^1_{f} \times Q_2, g^1_f) \cup (Q_1 \times Q^2_{f}, g^2_f)))$
  where the transition function $\Delta$ and the output function $\mu$ are defined as
  \begin{prooftree}
    \def\defaultHypSeparation{\hskip .05in}
    \AxiomC{$p \xrightarrow{a\mid b} q$ in $\Tt_1$,}
    \AxiomC{$p' \xrightarrow{a\mid c} q'$ in $\Tt_2$}
    \RightLabel{\scriptsize(1)}
    \BinaryInfC{$(p,p') \xrightarrow{a\mid (b,c)} (q,q')$  in $\Tt_1 \times \Tt_2$}
  \end{prooftree}
\end{definition}

Let us now examine the first step in more detail. The \GTTfull s that we construct are such that if the \GTT\ is for an atomic proposition, then it reads letters from $\Sigma$ and outputs a Boolean value. Otherwise, if the \GTT\ is for an operator (either Boolean or temporal), then it reads letters from $\{0,1\}$ (letters from $\{0,1\}^2$ respectively) if the operator is unary (binary respectively) and outputs Booleans. The heart of the construction is thus the following theorem, that we will prove partially here (for the case of non-temporal operators) leaving the proofs for temporal cases to next two sections.

\begin{figure}[!htbp]
  \centering
  \includegraphics[page=1]{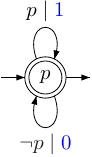}
  \hfil
  \includegraphics[page=2]{gastex-pictures-pics.pdf}
  \hfil
  \includegraphics[page=3]{gastex-pictures-pics.pdf}
  \hfil
  \includegraphics[page=4]{gastex-pictures-pics.pdf}
  \caption{Transducers $\Tt_{p}, \Tt_{\neg}, \Tt_{\wedge},\Tt_{\vee}$ for an atomic
  proposition $p$, and Boolean operators (negation, conjunction and disjunction,
  respectively).  Recall that $\Sigma \subseteq 2^{\Prop}$, where $\Prop$ is the set of
  atomic propositions.  Therefore, reading a $p$ in transducer $\Tt_{p}$ means that we
  read a letter $a$ such that $p \in a$.
  }
  \label{fig:basic-prop}
\end{figure}

\begin{theorem}~\label{thm:correctness}
  For each atomic proposition $p$ and each operator $\oper\in\{\neg,\lor,\land, \X_I, \U_I\}$,
  \begin{enumerate}
    \item if $\varphi=p$, with $\sem{p}:\TW\rightarrow \TB$, then there exists a safe
    \GTT\ $\Tt_p$ such that $\sem{\Tt_{p}}=\sem{\varphi}$.

    \item if $\varphi= ~\oper~p$, where $\oper$ is a unary operator, with
    $\sem{\varphi}:\TB\rightarrow \TB$, then there exists a safe \GTT\ $\Tt_{\oper}$ such that
    $\sem{\Tt_{\oper}}=\sem{\varphi}$.
    
    \item if $\varphi=p_\ell~ \oper ~p_r$, for atomic propositions $p_\ell$ and $p_r$, and
    $\oper$ a binary operator with $\sem{\varphi}:\TtwoB\rightarrow \TB$, then there
    exists a safe \GTT\ $\Tt_{\oper}$ such that $\sem{\Tt_{\oper}}=\sem{\varphi}$.
  \end{enumerate}
\end{theorem}

The proof of Part (1) is easy.  For any atomic proposition $p$, the \GTT\ $\Tt_p$ is shown in the leftmost picture of Figure~\ref{fig:basic-prop}.  We remark that in this transducer $\Tt_{p}$, a transition labelled $p$ (resp.\ $\neg p$) means that we read a letter $a$ such that $p \in a$ (resp.\ $p\notin a$).  For instance, let $\Prop = \{p,q\}$ and $\Sigma = \{a,b,c,d\}$ where the letters are interpreted as $a=\emptyset$, $b=\{p\}$, $c=\{q\}$ and $d=\{p,q\}$.  Then, in $\Tt_{p}$, the transition $p$ is taken when $b$ or $d$ is seen, and $\neg p$ is taken on reading $a$ or $c$.

Next, for the Boolean unary operator $\neg$ and binary operators $\land$ and $\vee$, the respective transducers are also depicted in Figure~\ref{fig:basic-prop}.  Since these are not temporal operators, the \GTT s use no clocks.
Also, these \GTT\ are deterministic and hence functional.

Now to complete the proof of the above theorem, it only remains to show Parts (2) and (3)
for the temporal operators in the timed setting.

In the rest of this section we will assume these results and discuss how the transducers
for the subformulae can be combined to get the transducer for a given formula.  We accomplish this by using two standard constructions for transducers, namely composition and
product.  For two transducers $\Tt$ and $\Tt'$, their (sequential) \emph{composition} is
denoted $\Tt\circ \Tt'$ and their \emph{product} is denoted $\Tt\times\Tt'$.

Suppose we have an $\mitl$ formula $\varphi \equiv \oper ~\psi$, and suppose we have the
\GTT\ $\Tt_\psi$ for $\psi$.  By Theorem~\ref{thm:correctness}, we also have a \GTT\
$\Tt_{\oper}$ for $\oper$.  Then, the \GTT\ for $\varphi$ is obtained by sequentially
composing the \GTT\ for $\oper$ and the \GTT\ for $\psi$, i.e.,
$\sem{\varphi}=\sem{\Tt_{\oper}\circ \Tt_\psi}$.  Similarly, for an $\mitl$ formula
$\varphi \equiv \psi_{\ell} ~\oper~ \psi_r$, if we have the \GTT\ $\Tt_\ell$ for $\psi_{\ell}$ and $\Tt_r$ for $\psi_r$.  Again, by Theorem~\ref{thm:correctness}, for any binary operators, we have a \GTT\ $\Tt_{\oper}$ for $\oper$.  Then, the \GTT\ for $\varphi$ is obtained by taking the product of $\Tt_{\ell}$ and $\Tt_{r}$ and then composing with $\Tt_{\oper}$, i.e., $\sem{\varphi}=\sem{\Tt_\oper \circ (\Tt_\ell\times \Tt_r)}$.  
Thus we have the following lemmas.

\begin{restatable}{lemma}{lemcomposition}
  \label{lem:composition}
  If $\Tt_1$ and $\Tt_2$ are functional \GTT s from alphabets $A$ to $B$ and $B$ to $C$ respectively, then $\Tt_2 \circ \Tt_1$ is a functional \GTT\ from $A$ to $C$ and $\sem{\Tt_{2}\circ\Tt_{1}}=\sem{\Tt_{2}}\circ\sem{\Tt_{1}}$.
\end{restatable}

\begin{restatable}{lemma}{lemproduct}
  \label{lem:product}
  If $\Tt_1$ and $\Tt_2$ are functional \GTTfull s from alphabets $A$ to $B$ and $A$ to $C$ 
  respectively, then $\Tt_1 \times \Tt_2$ is a functional \GTT\ from $A$ to $B\times C$ and 
  for each $(w,\tau)\in\dom{\Tt_{1}}\cap\dom{\Tt_{2}}$, if 
  $\sem{\Tt_{1}}((w,\tau))=(b_{0},\tau_{0})(b_{1},\tau_{1})\cdots$ and
  $\sem{\Tt_{2}}((w,\tau))=(c_{0},\tau_{0})(c_{1},\tau_{1})\cdots$ then
  $\sem{\Tt_{1}\times\Tt_{2}}((w,\tau))=((b_{0},c_{0}),\tau_{0})((b_{1},c_{1}),\tau_{1})\cdots$.
\end{restatable}

Finally, we can prove our main translation Theorem~\ref{thm:main-trans} below.

\begin{restatable}{theorem}{maintrans}
  \label{thm:main-trans}
  For any $\mitl$ formula, $\varphi$, we can construct a functional $\GTT$ $\Tt_\varphi$ such that $\sem{\varphi}=\sem{\Tt_\varphi}$.
\end{restatable}
  
The proof proceeds by structural induction on $\varphi$, and follows directly from the correctness of individual \GTT\ (given by Theorem~\ref{thm:correctness}) and the correctness and functionality of the composition and product constructions (given by Lemmas~\ref{lem:composition} and \ref{lem:product}). Finally, from this \GTT\ we obtain a safe \GTA\ that recognises the language $\lang(\varphi)$, by just looking at the output on first letter of input. Safety of this \GTA\ follows by the fact that we ensure safety at every step of the construction.

\begin{restatable}{corollary}{maincoraut}
  \label{cor:main-cor}
  For any $\mitl$ formula, $\varphi$, we can construct a safe $\GTA$ $\Aa_\varphi$ such that $\lang(\varphi)=\lang(\Aa_\varphi)$.
\end{restatable}

\subsection{Translation of Next and Until operators}~\label{sec:mitl-gta-future}

Now, we will describe the transducers for the modalities $\X_I$
and $\U_I$. 
We next describe the transducers for $\Next_I$ and $\Until_I$.

\subparagraph*{Next operator.} The transducer for $\Next_{I} p$ is given in
Figure~\ref{fig:Next-timed}.  It is obtained by extending the untimed variant of the
Next-transducer with a future clock $x$ that predicts the time to the next event.  The
idea is the same as explained in Figure~\ref{fig:intro-example} of the Introduction.  The
prediction of the next event is verified, by having the guard $x = 0$ in every transition.
Notice the use of the program syntax in this example: a transition first checks if $x = 0$
(satisfying a previous obligation), and then releases $x$ to a non-deterministic value
guessing the time to the next event, and then asks for a guard, either $-x \in I$ or $-x
\notin I$.  

\begin{figure}[tb]
  \centering 
  \includegraphics[page=5,scale=0.9]{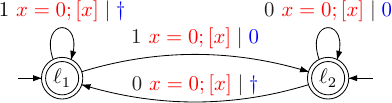}
  \caption{GTT for $\X_{I}$ using a future clock $x$. The output
    is $0$ for transitions to location $\ell_{2}$ and is the
    if-then-else \textcolor{blue}{$\dagger=\ite{(-x\in I)}{1}{0}$} for
    transitions to location $\ell_{1}$, where $I$ is some interval.}
  \label{fig:Next-timed}
\end{figure}

\subparagraph*{Until operator.}
We start by describing the transducer for the untimed $\U$ modality
$p \U q$ (in other words, $p \U_I q$ with $I = [0, \infty)$). This is
shown in Figure~\ref{fig:ltl-until}. For simplicity, we have assumed
$\Prop = \{p, q\}$ and the alphabet is
represented as $(0,0), (1, 0), (0,1), (1,1)$ corresponding to $\{\}$,
$\{p\}$, $\{q\}$ and $\{p, q\}$. On the word $w$, if $s_i$ is the
state that reads $a_i$, then the following invariants hold:
\begin{itemize}
\item $s_i = q$ iff $q \in a_i$,
\item $s_i = \neg q \land (p \U q)$ iff $q \notin a_i$ and
  $(w, i) \models p \U q$,
\item $s_i = \neg (p \U q)$ iff $(w, i) \not \models p \U q$.
\end{itemize}

\begin{figure}[tb]
  \centering
  \includegraphics[page=6,scale=0.9]{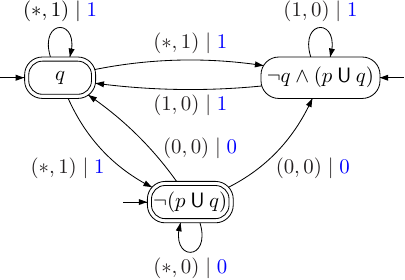}
  \hfill
  \includegraphics[page=7,scale=0.9]{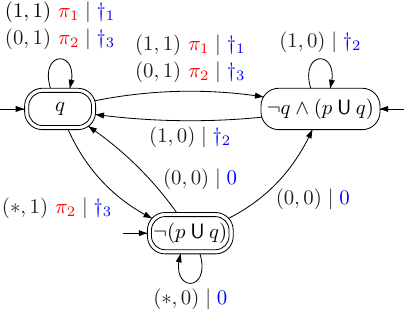}
  \caption{(Left) Transducer for the untimed LTL operator $p \U q$.\\
    (Right) Transducer $\Aa$ tracking the earliest and last $q$
    witnesses for $p \U q$. Program $\textcolor{red}{\pi_1}$ is
    $\textcolor{red}{x = 0; [x]}$ and program $\textcolor{red}{\pi_2}$
    is $\textcolor{red}{x = 0; y = 0; [x,y]}$. The outputs $\dagger_i$
    depend on interval $I$ in the timed until $p\U_{I}q$.}
  \label{fig:ltl-until}
  \label{fig:until-x-y}
\end{figure}

At the initial state the automaton makes a guess about position $0$,
and then subsequently on reading every $a_i$, it makes a guess about
position $i + 1$ and moves to the corresponding state.  The
transitions implement this guessing protocol. For instance,
transitions out of state $q$ read letters with $q = 1$, and also
output $1$; transitions out of state $\neg (p \U q)$ have output
$0$. A noteworthy point is that state $q \land \neg (p \U q)$ is
non-accepting, preventing the automaton to stay in that state
forever. For every word, the transducer has a unique accepting run and
the output at position $i$ is $1$ iff $(w, i) \models p \U q$.

Let us move on to the timed until $\U_I$. Let us forget the
specific interval $I$ for the moment. We will come up with a generic
construction, on which the outputs can be appropriately modified for
specific intervals. To start the construction, we need the following
notion.

\begin{definition}
  Let $w = (a_0, t_0) (a_1, t_1) \dots$ be a timed word and let
  $i \ge 0$.  The \emph{earliest} $q$-witness at position $i$ is the
  least position $j > i$ such that $q \in a_j$, if it exists. We
  denote this position $j$ giving the earliest $q$-witness at $i$ as
  $i_f$.  The \emph{last} $q$-witness is the \emph{least} position
  $j > i$ that satisfies
  \begin{align*}
    \alpha = q \wedge \neg(p \wedge \X(p\U q)) \equiv (q \land \neg p)
    \lor (q \land \X \neg (p \U q)) 
  \end{align*}
  We denote this position $j$ giving the last $q$-witness at $i$ as
  $i_\ell$.
\end{definition}

The earliest and last $q$-witnesses provide a convenient mechanism to
check $p \U_I q$ which, in many cases, can be deduced by knowing the
time to the earliest and last witnesses. Figure~\ref{fig:x-y-idea}
illustrates the interpretation of $x$ and $y$.

\begin{figure}[tbh]
  \centering
  \begin{tikzpicture}
    \draw [thin, gray] (0.7,0) to (7.2,0); \foreach \x in {1, 2.5, 3,
      4.7, 6.1, 7} \draw (\x, -0.1) to (\x, +0.1);
     
    \foreach \x in {2.5, 3, 4.7, 6.1} \node at (\x, +0.2) {\small
      $q$};

    \node at (1, -0.3) {\small $i$}; \node at (2.5, -0.3)
    {\small $i_f$}; \node at (7, +0.2) {\small $\alpha$};
    \node at (7, -0.3) {\small $i_\ell$};

    \begin{scope}[very thin, gray, yshift=0.3cm]
      \draw (1, 0.3) to node [above] {\small $\neg q$}
      (2.45,0.3); \draw (2.55, 0.3) to node [above] {\small
        $\neg q$} (2.95, 0.3);

      \draw (3.05, 0.3) to node [above] {\small $\neg q$} (4.65,
      0.3); \draw (4.75, 0.3) to node [above] {\small $\neg q$}
      (6.05, 0.3);

      \draw (6.15, 0.3) to node [above] {\small $\neg q$} (6.95,
      0.3);

      \draw (1, 0.8) to node [above] {\small $\neg \alpha$}
      (6.95, 0.8);
    \end{scope}
    \begin{scope}[yshift=-0.2cm, very thin, gray]
      \draw (1, -0.5) to node [below,black] {\small $x$} (2.45,
      -0.5);
       
      \draw (1, -0.9) to node [below,black] {\small $y$}
      (7,-0.9);
    \end{scope}
  \end{tikzpicture}
  \caption{Division of $q$ events, and interpretation of $x,y$}
  \label{fig:x-y-idea}
\end{figure}
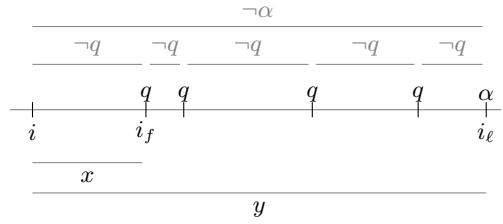

Our next task is to extend the $\U$ transducer of
Figure~\ref{fig:ltl-until} to include two future clocks $x$ and $y$
that predict at each $i$, the time to $i_f$ and $i_\ell$,
respectively. Figure~\ref{fig:until-x-y} describes the transducer
$\Aa$. For clock $x$ to maintain time to $i_f$ at each position $i$,
we can do the following: at every transition that reads $q$, the
transducer checks for $x = 0$ as guard and releases $x$ (with the time
to the next $q$).  If there is no such $q$, then $x$ needs to be
released to $-\infty$ in order to continue the run, as our timed words
are non-Zeno. Transitions satisfying $\neg q$ do not check for a guard
on $x$ or release $x$. Therefore, in any run, the value of $x$
determines the time to the next $q$ event.

In Figure~\ref{fig:x-y-idea}, the last witness (property $\alpha$) can
be identified by transitions of the form $(0, 1)$ (signifying
$q \land \neg p$) and transitions $(*, 1)$ going to state
$\neg (p \U q)$ (for $q \land \X \neg (p \U q)$). Similar to the
previous case of the earliest witness, every time we see such a
transition we check for $y = 0$ as a guard and release $y$. No other
transition checks or updates $y$. 
Notice that only the transitions with $q$ have been changed. All
transitions $(0, 1)$ check and release both clocks (program
$\textcolor{red}{\pi_2}$). Transitions $(1,1)$ that do not go to
$\neg (p \U q)$ check and release only $x$ (program
$\textcolor{red}{\pi_1}$), whereas the $(*,1)$ transition that goes to
$\neg (p \U q)$ does $\textcolor{red}{\pi_2}$.

\begin{lemma}
  For every timed word $w = (a_0, \tau_0) (a_1, \tau_1) \cdots$, there
  is a unique run of $\Aa$ of the form:
  $(s_0, v_0) \xra{\tau_0, \theta_0} (s_1, v_1) \xra{\tau_1 - \tau_0,
    \theta_1} \cdots$ such that for every position $i \ge 0$: (1)
  $s_i$ is state $q$ of $\Aa$ iff $w, i \models q$, (2) $s_i$ is state
  $\neg q \land (p \U q)$ iff $w, i \models \neg q \land (p \U q)$,
  (3) $s_i$ is state $\neg (p \U q)$ iff $w, i \models \neg (p \U q)$,
  (4) $v_i(x) = \tau_i - \tau_{i_f}$ and
  $v_i(y) = \tau_i - \tau_{i_\ell}$.
\end{lemma}

Using $\Aa$ we can already answer $p \U_I q$ for one-sided
intervals: $[0, c], [0, c), [b, +\infty), (b, +\infty)$, for natural
numbers $b, c$.
\begin{itemize}
\item if $0 \in I$: $\dagger_1 = \dagger_3 = 1$ (current
  position is a witness), and
  $\dagger_2 = \ite{(-x \in I \lor -y \in I)}{1}{0}$,
\item if $0 \notin I$: $\dagger_3 = 0$, and
  $\dagger_1 = \dagger_2 = \ite{(-x \in I \lor -y \in I)}{1}{0}$.
\end{itemize}
This is because in one-sided intervals, if at all there is a witness,
the earliest or the last is one of them.

\subparagraph*{Until with a non-singular interval.} 
We will now deal with the case of intervals $I=[b, c]$ with
$0 < b < c < \infty$. Firstly, using $x$ and $y$, some easy cases of
$p \U_I q$ can be deduced. Output remains $0$ for transitions starting
from $\neg(p\U q)$.
For other transitions, here are some extra checks:
\begin{itemize}
\item if $-x \in I$ or $-y \in I$, output $1$ (one of the earliest or
  last witness is also a witness for $p \U_I q$),
\item else, if $-y < b$ or $c < -x$, output $0$ (the time to the last
  witness is too small or the time to the earliest witness is too
  large, so there is no witness within $I$).
\end{itemize}

If neither of the above cases hold, then we need guess a potential
witness within $[b,c]$ and verify it. This requires substantial
book-keeping which we will now explain. Assume we are given a timed
word $w = (a_0, t_0) (a_1, t_1) \cdots$.
Let us a call $j\geq0$ a \emph{difficult point} if:
\begin{align*}
  w, j \models p \U q \text{ and } t_{j_f} < t_j + b \text{ and } t_j
  + c < t_{j_\ell} 
\end{align*}
This leaves the possibility for a $q$-witness within $[b, c]$. So, for
difficult points, we need to make a prediction whether we have a
$q$-witness within $[t_j+b,t_j+c]$: guess a time to a witness within
$[t_j+b,t_j+c]$ and check it. We cannot keep making such predictions
for every difficult point as we have only finitely many
clocks. Therefore, we will guess some special witnesses. First we
state a useful property.

\begin{lemma}
  Let $j$ be a difficult point. Then, for all $k$ such that
  $j \le k \le j_\ell$, we have $w, k \models p \U q$.
\end{lemma}
Therefore, automaton $\Aa$ stays in the top two states, while reading
$j$ upto $j_\ell$.

\begin{figure}[tbh]
  \centering
  \begin{tikzpicture}
    \draw [thin, gray] (0,0) -- (8.5,0); \draw (0,0) to (0, 0.2);
    \node at (0, 0.4) {\small $t_j$};
    \draw (7,0) to (7, 0.2); \node at (7, 0.4) {\small
    $t_j + c$}; \draw (6,0) to (6, -0.2); \fill[red] (6, 0) circle
    (1.5pt); \node at (6, -0.4) {\small $t_{j'}$}; \draw (1,0)
    to (1, -0.2); \node at (1, -0.4) {\small $t_{j'} - b$};
    \draw [red] (0,-0.1) to (1,-0.1); \draw (8.4, 0) to (8.4, -0.2);
    \fill [blue] (8.4,0) circle (1.5pt); \node at (8.4, -0.4)
    {\small $t_{j''}$};
    
    \draw [blue] (1.4, 0.1) to (3.4, 0.1); \draw (1.4, 0) to (1.4,
    0.2);
    \draw (3.4,0) to (3.4, 0.2);
    \node at (1.4, 0.4) {\small $t_{j''} - c$};
    \node at (3.4, 0.4) {\small $t_{j''} - b$};
  \end{tikzpicture}
  \caption{Illustration of a point $j$. The point $j'$ is the last
  $q$-witness before $t_j+c$, and $j''$ is the first $q$-witness
  after $t_j+c$.}
  \label{fig:Until-idea}
\end{figure}
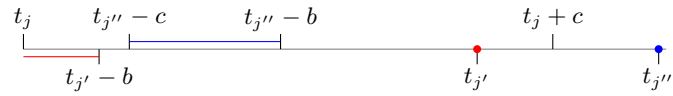

We will now come back to the idea of choosing special witnesses.  This is illustrated in
Figure~\ref{fig:Until-idea}.  For a point $j$, we let $j' \ge j$ be the greatest position
containing $q$ such that $t_{j'} \le t_j + c$.  Let $j'' > j$ be the least position
containing $q$ such that $t_j + c < t_{j''}$.  So, no position $j' < k < j''$ contains
$q$.  While reading a difficult point $j$, let us make use of fresh clocks $x_1$ and $y_1$
to predict these two witnesses:
\begin{align*}
  x_1 = t_j - t_{j'} \qquad y_1 = t_j - t_{j''}
\end{align*}

  For the next important observation, we will once again take the help
  of Figure~\ref{fig:Until-idea}. Notice that for all points $i$ with
  $t_i \in [t_j, t_{j'}-b]$, the point $j'$ is also a witness for
  $w, i \models p \U_{[b,c]} q$. Similarly, $j''$ is a witness for all
  $i$ such that $t_i \in [t_{j''} - c, t_{j''} - c]$.  Therefore for
  all $i$ such that $t_i \in [t_j, t_{j''} - b]$, we have a way to
  determine the output: it is $1$ iff while reading $a_i$ we have
  $-x_{1}\in I$ or $-y_{1}\in I$ (recall we have predicted $x_1$ and
  $y_1$ while reading $a_j$ as explained above).
  So, we do not have to make new guesses at the difficult points in
  $[t_j, t_{j''} - b]$. After $t_{j''} - b$ (which can be identified
  with the constraint $-b < y_1$), we need to make new such guesses,
  using fresh clocks, say $x_2, y_2$. We will call the difficult
  points where we start new guesses as \emph{special difficult
    points}. Notice that the distance between two special difficult
  points is at least $c - b$ (which is $\ge 1$, as we consider
  non-singular intervals with bounds in $\mathbb{N}$).  In the figure,
  if $j$ is a special point, a new special point will be opened later
  than $t_{j''} - b$.

  This gives a bound on the number of special points that can be open
  between $j$ and $j''$. Suppose
  $j < \ell_1 < \ell_2 < \dots < \ell_i < j''$ be the sequence of
  special points between $j$ and $j''$. Since $\ell_1$ is opened when
  time to $j''$ is atmost $b$, we get the inequality:
  $t_{\ell_i} - t_{\ell_1} < b$. Since consecutive special points are
  at least $c - b$ apart, we have $(i - 1) (c - b) < b$. This entails
  $i < 1 + \lceil \dfrac{b}{c - b} \rceil$. By the time we reach
  $j''$, we need to have opened at most
  $k = 1 + \lceil \dfrac{b}{c - b} \rceil$ special points, and hence
  we can work with the extra clocks
  $x_1, y_1, x_2, y_2, \dots, x_k, y_k$.

\begin{figure}
  \centering
  \includegraphics[page=8,scale=0.9]{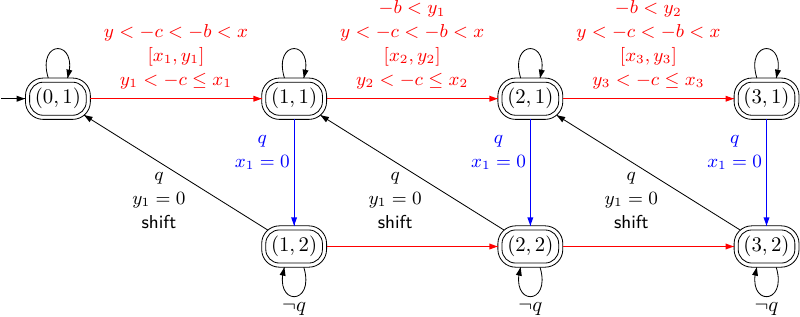}
  \caption{The automaton for predicting $q$-witnesses which are not
    given by the earliest and latest. For clarity, not every transition
    is indicated. All clocks are future clocks.}
  \label{fig:book-keeping-until}
\end{figure}

All these ideas culminate in the definition of a book-keeping
automaton $\Bb$ \label{auto:book-keeping-until} to handle difficult
points. This is shown in Figure~\ref{fig:book-keeping-until} for
$N = 3$. The automaton $\Bb$ synchronizes with $\Aa$ (via a usual
cross-product synchronized on transitions).  All transitions of $\Bb$
other than the self loop on state $(0, 1)$ satisfy $p$. Transitions
which satisfy $q$, and $\neg q$ are specifically marked in the figure.

The automaton $\Bb$ starts in the initial state $(0,1)$. It moves to
$(1, 1)$ on the first difficult point $j$ (which will become special)
and comes back to $(0, 1)$ when there are no active special difficult
points waiting for witnesses (a special difficult point $j$ is active
at positions $j \le i \le j''$).  States $(i, 1)$ in the top indicate
that there are $i$ active special difficult points currently. A state
$(i, 2)$ indicates that the $j'$ witness for the oldest active point
has been seen, and we are waiting for its $j''$ witness (the space
between $t_{j'}$ and $t_{j''}$ in Figure~\ref{fig:Until-idea}).

The red transitions $(i, *) \to (i+1, *)$ open new special difficult
points, and contain the program as illustrated in the figure. At
$(i, 1)$, suppose $\ell_1 < \ell_2 < \cdots < \ell_i$ are the active
special difficult points where we have predicted
$x_1, y_1, \dots, x_i, y_i$ respectively. We have the invariant:
$$y_i < x_i \le y_{i-1} < x_{i-1} \le \cdots \le y_2 < x_2 \le y_1 < x_1 $$ 
Notice that we may have $x_i=y_{i-1}$: the ``first'' witness of the
$i^{th}$ special point ($\ell_i'$) could coincide with the ``second''
witness of the $(i-1)^{th}$ point ($\ell_{i-1}'$). This leads to
certain subtleties, which we will come to later.

The blue transitions read the witness for the oldest active special
point (that is, we have reached $\ell_1'$). Observe that $x_1 = 0$
does not immediately identify $\ell_1'$, since there could be a
sequence of positions at the same time, and $\ell'_1$ is the last of
them. Therefore, we make a non-deterministic choice whether to take
the blue transition (implying that $\ell'_1$ has been found), or we
remain in the same state.
The blue transitions read a $q$, check $x_1 = 0$, and then releases
$x_1$ to $-\infty$ (not shown in Figure~\ref{fig:Until-idea}).  The
black (diagonal) transitions witness $\ell''_1$. When this happens,
$x_1, y_1$ are no longer useful, and therefore all the higher clocks
are shifted using the permutation \textsf{shift} which maps
$x_2, y_2, \dots x_k, y_k, x_1, y_1$ to $x_1, y_1, \dots, x_k, y_k$
and keeps the other clocks unchanged.

There are some subtleties which arise when special points coincide
with witness points, or when the second witness of a special point
coincides with the first witness of the consecutive special point, which we discuss next. 
These subtleties do not increase the states, but they result in more
transitions. 

\subparagraph*{Subtleties.} The
first subtlety arises when we have $\ell''_j = \ell'_{j+1}$ for
consecutive special points.  This will imply $y_j=x_{j+1}$. The
reverse direction is not true, as there could be a sequence of
positions with the same time, but let us assume we have dealt with it
by the non-deterministic choice. When we actually witness these
points, the clock values would have shifted to lower indices.  This
situation will be manifested as $y_1 = x_2=0$. Suppose we are in
$(i, 2)$ and see a point $\ell_j''$ ($y_1 = 0$). The diagonal
transition takes the automaton to $(i-1, 1)$ and shifts $x_2$ to
$x_1$. Now, $x_1 = 0$ (as $\ell'_{j+1} = \ell''_j$). Therefore, we
will have to combine the black-diagonal-left with the downward-blue to
get the combined effect. This leads to these two divisions:
\begin{align*}
  (i,2) &\xra{ y_1 = 0 } (i-1, 1) & 
  (i,2) &\xra{ y_1 = 0 \land x_2 = 0 } (i - 1, 2)
\end{align*}

The second subtlety is that one of either $\ell_j'$ or $\ell''_j$
witnesses be a new special point (notice that the red transitions are
independent of the blue and black transitions). In such cases, we can
combine the two 
effects in any order: first discharge $x_1$ or $y_1$ verification, and
then open a new special point or vice-versa. This leads to some
additional divisions of the form:
\begin{align*}
  (i, 1) & \xra{x_1 = 0 \land -b < y_i} (i+1,2) &
  (i, 2) & \xra{y_1 = 0 \land -b < y_i} (i, 1)
\end{align*}
In the first transition, we have combined a blue and a red (in any order); whereas in
the second, we have combined a red and a black-diagonal, in any order.

The third subtlety is that the first and second subtleties may occur
together! A point could be $\ell''_j$, $\ell'_{j+1}$ and also a
new special point. We illustrate this on a specific state $(i, 2)$. We
provide only the ``guards''. The full program is obtained by 
suitably combining the effects of the individual transitions:
\begin{align*}
  (i, 2) & \xra{ y_1 = 0  \land y_i \le -b } (i-1, 1)
  &
  (i, 2) & \xra{ y_1 = 0  \land -b < y_i} (i, 1) 
  \\
  (i, 2) & \xra{ y_1 = 0 \land x_2 = 0 \land y_i \le -b} (i - 1, 2)
  &
  (i, 2) & \xra{ y_1 = 0 \land x_2 = 0 \land -b < y_i} (i, 2)
\end{align*}

We conclude by giving the full construction of $\Bb$, taking into account all
these subtleties described above.
The set of states is $\{0,1,\ldots,N\}\times\{1,2\}$
where $N=1 + \lceil \dfrac{b}{c-b} \rceil$ (actually state $(0,2)$ is not reachable).  The initial state is $(0,1)$ and all states are accepting.

\subparagraph*{Final description.}
The transitions are described in pseudo-code in Algorithm~\ref{algo:auto-B}.  
In comments, we use the terminology introduced before: difficult point, special difficult
point; and we also refer to the color of transitions in
Figure~\ref{fig:book-keeping-until}.

Parsing the pseudo-code from a current state ($(k,m)$ results in a sequence of guards and
releases, an output of a boolean value (\Output value), and the next state (\goto
$(k',m')$).  The most difficult case is for states $(k,2)$ with $k\geq2$, where we could
generate transitions to states $(k-1,1),(k-1,2),(k,1),(k,2),(k+1,2)$.

\begin{algorithm}[!htbp]
  \small
\caption{Automaton $\Bb$ (synchronized with $\A$)}
\label{algo:auto-B}
\begin{algorithmic}[1]

  \Case{State $(0,1)$} \Comment initial state of $\Bb$
    \IfThen{$\Aa$ at state $\neg (p \U q)$}{\Output 0; \goto (0,1)}\EndIfThen
    \IfThen{$-x\in I$ or $-y\in I$}{\Output 1; \goto (0,1)}\EndIfThen
    \IfThen{$x<-c$ or $-b<y$}{\Output 0; \goto (0,1)}\EndIfThen
    \State Release $[x_{1},y_{1}]$ \Comment Special difficult point
    \State Check $y_{1}<-c\leq x_{1}$ 
    \State \Output $(x_{1}\leq -b)$ \Comment boolean value 
    \State \goto (1,1) \Comment red transition
  \EndCase
  
  \Case{State $(k,1)$ with $k>0$} \Comment waiting for the event predicted by $x_{1}$
    \State $k'\gets k$
    \If{$y_{k}\leq -b$}
      \State \Output $(x_{k}\in I) \vee (y_{k}\in I)$ \Comment boolean value 
    \Else
      \If{$-c\leq y$} \Comment not a difficult point
        \State \Output $(y\leq -b)$ \Comment boolean value ($y\in I$)
      \Else  \Comment new special difficult point, red transition, 
        \State \Comment possibly combined with a blue transition below
        \State $k'\gets k+1$; 
        \State Release $[x_{k},y_{k}]$; Check $y_{k}<-c\leq x_{k}$ 
        \State \Output $(x_{k}\leq -b)$ \Comment boolean value 
      \EndIf
    \EndIf
    \Choose{}
      \When{True}{\goto $(k',1)$}   \Comment not the event predicted by $x_{1}$ \EndWhen
      \When{$q \wedge (x_{1}=0)$}{Release $[x_{1}]$; $x_{1}=-\infty$; \goto $(k',2)$} \EndWhen
      \Comment blue transition
    \EndChoose
  \EndCase
  
  \Case{State $(k,2)$ with $k>0$} \Comment waiting for the event predicted by $y_{1}$
    \State $k'\gets k$
    \If{$y_{k}\leq -b$}
      \State \Output $(x_{k}\in I) \vee (y_{k}\in I)$ \Comment boolean value 
    \Else
      \If{$-c\leq y$} \Comment not a difficult point
        \State \Output $(y\leq -b)$ \Comment boolean value ($y\in I$)
      \Else \Comment new special difficult point, red transition, 
        \State \Comment possibly combined with a black transition below
        \State $k'\gets k+1$;
        \State Release $[x_{k},y_{k}]$; Check $y_{k}<-c\leq x_{k}$ 
        \State \Output $(x_{k}\leq -b)$ \Comment boolean value 
      \EndIf
    \EndIf
    \If{$\neg q$} \Comment not the event predicted by $y_{1}$ 
      \State \goto $(k',2)$
    \Else \Comment event predicted by $y_{1}$, black transition
      \State \Comment possibly combined with a blue transition below
      \State Check $y_{1}=0$; Release $[y_{1}]$; $y_{1}=-\infty$
      \If{$k'=1$}
        \State \goto $(0,1)$
      \Else
        \State Shift $x_{2},y_{2},\ldots,x_{k},y_{k},x_{1},y_{1}$ to $x_{1},y_{1},\ldots,x_{k},y_{k}$
      \EndIf
      \Choose{}
        \When{True}{\goto $(k'-1,1)$} \EndWhen 
        \Comment not the event predicted by the new $x_{1}$ 
        \When{$(x_{1}=0)$}{$Release [x_{1}]$; $x_{1}=-\infty$; \goto $(k'-1,2)$} \EndWhen
        \Comment blue transition
      \EndChoose
    \EndIf
  \EndCase

\end{algorithmic}
\end{algorithm}

This concludes the description of the automaton
$\Bb$. The product $\Aa \times \Bb$ gives the required transducer for
$p \U_I q$.

\subparagraph*{Complexity and comparison with the MightyL approach.}
The final automaton $\Aa \times \Bb$ has at most $6k$ states, where
$k = 1 + \lceil \dfrac{b}{c - b} \rceil$ as defined above: automaton
$\Aa$ has 3 states, and automaton $\Bb$ has $2k - 1$ states (see
Figure~\ref{fig:book-keeping-until}). In terms of clocks, $\Aa$ has 2
future clocks $x, y$, and $\Bb$ has $2k$ future clocks
$x_1, y_1, \dots, x_k, y_k$. We have used a permutation operation
\textsf{shift}. As we mention in Remark~\ref{rem:renaming}, renamings can be
eliminated by maintaining in the current state the composition of
permutations applied since the initial state. Since each permutation
does a cyclic shift, in any composition, the clocks
$x_1, y_1, \dots, x_k, y_k$ are renamed to some
$x_i, y_i, \dots, x_k, y_k, x_1, y_1, \dots, x_{i-1},
y_{i-1}$. Therefore, there are at most $k$ renamings. Maintaining them
in states gives rise to atmost $\mathcal{O}(k^2)$ states.

In contrast, the state-of-the-art approach~\cite{MightyL} starts with
a 1-clock alternating timed automata for $\U_I$. After reading a timed
word, the 1-ATA reaches a configuration containing several
state-valuation pairs $(q, v)$. A finite abstraction of this set of
configurations, called the interval semantics, has been
proposed~\cite{BrihayeEG13,BrihayeEG14,MightyL}. This abstraction is
maintained in the states. Overall, the number of locations for
$p \U_I q$ is exponential in $k$, and the number of clocks is $2k+2$.

Due to the presence of future clocks, we are able to make predictions,
as in Figure~\ref{fig:Until-idea} and the GTA syntax enables concisely
checking these predictions in the transitions. Therefore, we are able
to give a direct construction to the final automaton, instead of going
via an alternating automaton and then abstracting it.

\section{Conclusion}

In this paper, we have answered two problems: (1) liveness of GTA and
(2) MITL model checking using GTA. The solution to the first problem
required to bypass the technical difficulty of having no finite
time-abstract bisimulation for GTAs. The presence of diagonal
constraints adds additional challenges. For MITL model checking
using GTA, we have described the GTA for the $\Next_I$ and
$\Until_I$ modalities. Indeed, the presence of future clocks allows
to make predictions better and we see an exponential gain over the
state-of-the-art, in the
number of states of the final automaton produced. Moreover, our
construction is direct, without having to go via alternation. 

The next logical step would be to implement these ideas and see how
they perform in practice, and compare them with existing
well-engineered tools (e.g.,~\cite{MightyL}). This will require a
considerable implementation effort, needing several optimizations
and incorporating of many practical considerations before it can
become scalable.  This provides tremendous scope for future work on these lines.

\clearpage

\bibliography{m.bib}

\end{document}

%% file: declarations.tex
\newcommand{\eqv}{\approx}
\newcommand{\Int}[1]{\lfloor #1 \rfloor}
\renewcommand{\frac}[1]{\{ #1 \}}

\newcommand{\Aa}{\mathcal{A}}
\newcommand{\Bb}{\mathcal{B}}

\newcommand{\Tt}{\mathcal{T}}
\newcommand{\Ii}{\mathcal{I}}

\newcommand{\TW}{T\Sigma^\omega}
\newcommand{\TB}{T\mathbb{B}^\omega}
\newcommand{\TtwoB}{(T\mathbb{B}^2)^\omega}
\newcommand{\oper}{\mathrm{op}}

\newcommand{\TWi}{T\Sigma^\omega}
\newcommand{\TWo}{T\Gamma^\omega}

\newcommand{\RR}{\mathbb{R}}
\newcommand{\Rpos}{\RR_{\ge 0}}
\newcommand{\Rneg}{\RR_{\le 0}}

\newcommand{\Nat}{\mathbb{N}}

\newcommand{\gzg}{\ensuremath{\mathsf{GZG}}}

\newcommand{\xra}[1]{\xrightarrow{#1}}

\renewcommand{\d}{\delta}
\newcommand{\s}{\sigma}

\newcommand{\A}{\mathcal{A}}

\newcommand{\leqlt}{\mathrel{\triangleleft}}

\newcommand{\da}{{\downarrow}}
\newcommand{\GG}{\mathbb{G}}

\newcommand{\V}{\mathbb{V}}

\newcommand{\elapse}[1]{\overrightarrow{#1}}

\newcommand{\prog}{\mathsf{prog}}
\newcommand{\guard}{\mathsf{guard}}
\newcommand{\chg}{\mathsf{change}}
\newcommand{\rnm}{\mathsf{rename}}

\newcommand{\perm}{\mathsf{perm}}

\newcommand{\Prog}{\textrm{Programs}}

\newcommand{\GTA}{\textrm{GTA}}

\newcommand{\GTAFull}{Generalized timed automata}
\newcommand{\GTAfull}{generalized timed automata}

\newcommand{\GTTFULL}{Generalized Timed Transducer}

\newcommand{\GTTfull}{generalized timed transducer}
\newcommand{\GTT}{\textrm{GTT}}

\newcommand{\TS}{\mathbb{TS}}

\newcommand{\mitl}{\mathsf{MITL}}

\newcommand{\Prop}{\mathsf{Prop}}

\newcommand{\X}{\mathop{\mathsf{X}\vphantom{a}}\nolimits}

\newcommand{\U}{\mathbin{\mathsf{U}}}

\newcommand{\F}{\mathop{\mathsf{F}\vphantom{a}}\nolimits}
\newcommand{\G}{\mathop{\mathsf{G}\vphantom{a}}\nolimits}

\newcommand{\Next}{\X}
\renewcommand{\Until}{\U}

\newcommand{\ite}[3]{{#1}\,?\,{#2}:{#3}}

\newcommand\sem[1]{[\![ #1 ]\!]}
\newcommand\dom[1]{\mathsf{dom}(#1)}
\newcommand{\tra}{\mathcal{T}}
\newcommand{\lang}{\mathcal{L}}